\documentclass[final]{siamltex}
\usepackage{amsmath,amsfonts,amssymb,graphicx,stmaryrd,color}

\def\RR{\mathbb{R}}

\def\ee{{(11)}}
\def\et{{(12)}}
\def\te{{(21)}}
\def\tt{{(22)}}

\title{Fixed-order H-infinity control for interconnected systems using delay differential algebraic equations}

\author{ Suat Gumussoy \and Wim Michiels\thanks{Department of Computer Science, Katholieke Universiteit Leuven, Belgium, (\texttt{\{Suat.Gumussoy,Wim.Michiels\}@cs.kuleuven.be})} }

\newtheorem{remark}[theorem]{Remark}

\newtheorem{example}[theorem]{Example}
\newtheorem{algorithm}[theorem]{Algorithm}
\newtheorem{assumption}[theorem]{Assumption}

\makeatletter
\def\Ddots{\mathinner{\mkern1mu\raise\p@
\vbox{\kern7\p@\hbox{.}}\mkern2mu
\raise4\p@\hbox{.}\mkern2mu\raise7\p@\hbox{.}\mkern1mu}}
\makeatother

\newcommand{\Hi}{{\cal H}_\infty}

\newcommand{\R}{\mathbb{R}}
\newcommand{\C}{\mathbb{C}}

\newcommand{\UU}{\mathbf{U}}
\newcommand{\VV}{\mathbf{V}}
\newcommand{\w}{\omega}

\gdef \RR{{\Bbb R}}

\gdef \ZZ{{\Bbb Z}}

\begin{document}
\maketitle

\begin{abstract}
We analyze and design H-infinity controllers for general time-delay systems with time-delays in systems' state, inputs and outputs. We allow the designer to choose the order of the controller and to introduce constant time-delays in the controller. The closed-loop system of the plant and the controller is modeled by a system of delay differential algebraic equations (DDAEs). The advantage of the DDAE modeling framework is that any interconnection of systems and controllers prone to various types of delays can be dealt with in a systematic way, without using any elimination technique. We present a predictor-correct algorithm for the H-infinity norm computation of systems described by DDAEs. Instrumental to this we analyze the properties of the H-infinity norm. In particular, we illustrate that it may be sensitive with respect to arbitrarily small delay perturbations. Due to this sensitivity, we introduce the strong H-infinity norm which explicitly takes into account small delay perturbations, inevitable in any practical control application. We present a numerical algorithm to compute the strong H-infinity norm for DDAEs. Using this algorithm and the computation of the gradient of the strong H-infinity norm with respect to the controller parameters, we minimize the strong H-infinity norm of the closed-loop system based on non-smooth, non-convex optimization methods. By this approach, we tune the controller parameters and design H-infinity controllers with a prescribed order or structure.
\end{abstract}

\section{Introduction}
In many control applications, robust controllers are desired to achieve stability and performance requirements under model uncertainties and exogenous disturbances \cite{zhou}. The design requirements are usually defined in terms of $\Hi$ norms of  closed-loop transfer functions including the plant, the controller and weights for uncertainties and disturbances. There are robust control methods to design the optimal $\Hi$ controller for linear finite dimensional multi-input-multi-output (MIMO) systems based on Riccati equations and linear matrix inequalities (LMIs), see e.g.~\cite{DGKF, GahinetApkarian_HinfLMI} and the references therein. The order of the controller designed by these methods is typically larger or equal then the order of the plant. This is a restrictive condition for high-order plants, since low-order controllers are desired in a practical implementation. The design of fixed-order or low-order $\Hi$ controller can be translated into a non-smooth, non-convex optimization problem.  Recently fixed-order $\Hi$ controllers have been successfully designed for finite dimensional linear-time-invariant (LTI) MIMO plants using a direct optimization approach \cite{suatHIFOO}. This approach allows the user to choose the controller order and tunes the parameters of the controller to minimize the $\Hi$ norm under consideration. An extension to a class of retarded time-delay systems has been described in \cite{bfgbookchapter}.

In this work we design a fixed-order or fixed-structure $\Hi$ controller in a feedback connection with a time-delay system. The closed-loop system is a delay differential algebraic system and its state-space representation is written as
\begin{equation}\label{system}
\left\{\begin{array}{l}
E \dot x(t)= A_0 x(t)+\sum_{i=1}^m A_i x(t-\tau_i) +B w(t), \\
z(t)=C x(t).
\end{array}\right.
\end{equation}
The time-delays $\tau_i$, $i=1,\ldots,m$ are positive real numbers and the capital letters are real-valued matrices with appropriate dimensions. The input $w$ and output $z$ are disturbances and signals to be minimized to achieve design requirements and some of the system matrices include the controller parameters.

The system with the closed-loop equations (\ref{system}) represents all interesting cases of the feedback connection of a time-delay plant and a controller. The transformation of the closed-loop system to this form can be easily done by first augmenting the system equations of the plant and controller. As we shall see, this augmented system can subsequently be brought in the form (\ref{system}) by introducing slack variables to eliminate input/output delays and direct feedthrough terms in the closed-loop equations. Hence, the resulting system of the form (\ref{system}) is obtained directly without complicated elimination techniques that may even not be possible in the presence of time-delays.

%Our approach for the fixed-order $\Hi$ controller design is based on a non-smooth, non-convex optimization method \cite{suatHIFOO} and tunes the controller parameters in the system matrices of (\ref{system}) in order minimize the $\Hi$ norm of (\ref{system}). This approach requires the computation of $\Hi$ norms  for DDAEs and their derivatives with respect to the controller parameters.

As we shall see, the $\Hi$ norm of DDAEs may be sensitive to arbitrarily small delay changes. Since small modeling errors are inevitable in any practical design we are interested in the smallest upper bound of the $\Hi$ norm that is insensitive to small delay changes. Inspired by the concept of strong stability of neutral equations \cite{have:02}, this leads us to the introduction of the concept of  \emph{strong $\Hi$ norms} for DDAEs,  Several properties of the strong $\Hi$ norm are shown and a computational formula is obtained. The theory derived can be considered as the dual of the theory of strong stability as elaborated in \cite{fridman2,have:02,extraE,TW-report-286,Michiels:2005:NEUTRAL,Michiels:2007:MULTIVARIATE} and the references therein.

 In addition, a level set algorithm for computing strong $\Hi$ norms is presented.
 %see~\cite{byers,boydbala,steinbuch} for the underlying idea behind level set %methods.
 Level set methods  rely on the property that
the frequencies at which a singular value of the transfer function equals a given value (the level) can be directly obtained from the solutions of a linear eigenvalue problem with Hamiltonian symmetry (see, e.g.~\cite{boydbala,boydbala2,steinbuch,byers}), allowing a two-directional search for the global maximum. For time-delay systems this eigenvalue problem is infinite-dimensional.
%In the article~\cite{wimsimax} on time-delay systems of retarded type, this problem is overcome with a predictor-corrector approach: in the prediction (approximation) step the infinite-dimensional problem is discretized allowing to apply methods for LTI systems, while in the correction step the effect of the approximation on the computed $\Hi$ norm is removed.
 %
 %Since time-delay systems are inherently infinite-dimensional systems
 Therefore, we adopt a predictor-corrector approach, where the prediction step involves a finite-dimensional approximation of the problem, and the correction serves to remove the effect of the discretization error on the numerical result. The algorithm is inspired by the algorithm for $\Hi$ computation for time-delay systems of retarded type as described in~\cite{wimsimax}. However, a main difference lies in the fact that the robustness w.r.t.~small delay perturbations needs to be explicitly addressed.

The  numerical algorithm for the norm computation is subsequently applied to the design of  $\Hi$ controllers by a direct optimization approach. In the context of control of  LTI systems it is well known that  $\Hi$ norms are in general  non-convex  functions of the controller parameters which arise as elements of the closed-loop system matrices. They are typically even not everywhere smooth, although they are differentiable almost everywhere~\cite{suatHIFOO}. These properties carry over to the case of strong $\Hi$ norms of DDAEs under consideration. Therefore, special optimization methods for non-smooth, non-convex problems are required. We will use
%
%
%his precludes the use of standard
%optimization methods, and requires
%Instead, nonsmooth optimization methods are appropriate such as the gradient sampling algorithm, \cite{OVERTON}. The latter is a first order method, that generalizes the steepest descent method to nonsmooth problems and is based on approximating the Clarke subdifferential or generalized gradient in each iterate.
%
% Recently, it has been shown by \cite{overtonbfgs} that the BFGS method can also be effectively applied to nonsmooth functions. It is a quasi-Newton method that approximates the Hessian by rank one updates based on first-order information. Since the BFGS algorithm is less computationally demanding and much faster it is advised to use it in the first place, and only switch to the gradient algorithm upon stagnation.}
a combination of BFGS, whose favorable properties in the context of non-smooth problems have been reported in \cite{overtonbfgs}, bundle and gradient sampling methods, as implemented in the MATLAB code HANSO\footnote{Hybrid Algorithm for Nonsmooth Optimization, see~\cite{overtonhanso}}. The overall algorithm only requires the evaluation of the objective function, i.e.,~the strong $\Hi$ norm, as well as its derivatives with respect to the controller parameters whenever it is differentiable. The computation of the derivatives is also discussed in the paper.

The presented method is frequency domain based and builds on the eigenvalue based framework developed in~\cite{bookwim}.  Time-domain methods for the $\Hi$ control of DDAEs have been described in \cite{fridman} and the references therein, based on the construction of Lyapunov-Krasovskii functionals.

\smallskip

The structure of the article is as follows. In Section \ref{sec:motex} we illustrate the generality of the system description (\ref{system}). Preliminaries and assumptions are given in Section \ref{sec:prelim}. The definition and properties of the strong $\Hi$ norm of DDAE are given in Section \ref{sec:shinf}. The numerical algorithm to compute the strong $\Hi$ norm is described in detail in Section \ref{sec:comp_shinf}. Fixed-order $\Hi$ controller design is addressed in Section \ref{sec:design}. Section~\ref{sec:ex} is devoted to the numerical examples. In Section \ref{sec:conc} some concluding remarks are presented. Some technical lemmas and finite dimensional approximation of time-delay systems are given in Appendices~\ref{sec:Appendix2} and \ref{sec:findimapp} respectively.

\subsection*{Notations} The notations are
as follows:
\begin{tabbing}
  \= $j$\hspace{1.5cm} \=:  the imaginary unit \\
  \> $\C, \R$ \>: set of the complex and real numbers \\
  \> $\mathbb{N}$\>: set of natural numbers \\
  \> $\R^+,\R_0^+$ \> : set of  nonnegative and strictly positive real numbers \\
  \> $A^{*}$ \>: complex conjugate transpose of the matrix $A$ \\
  \> $A^{-T}$ \>: transpose of the inverse matrix of $A$ \\
%  \> $A^{-T}$ \>: transpose of the inverse matrix of $A$ \\
  \> $A^{\bot}$ \>: matrix of full column rank whose columns  span \\
  \> \> \ \  the
   orthogonal complement of $A$   \\
  \> $I,I_n$ \> : identity matrix of appropriate dimensions, of dimensions $n\times n$ \\
  \> $0,0_n,0_{n\times m}$ \>: zero matrix with appropriate dimensions, with dimension $n \times n$,  \\
  \>  \>\ \ with dimensions $n \times m$ \\
  %\> $\sigma_{1}(A)$ \>: the largest singular value of the matrix $A$ \\
  \> $\sigma_i(A)$ \>: i$^\mathrm{th}$ singular value of $A$,\ $\sigma_1(\cdot)\geq\sigma_2(\cdot)\geq \cdots$ \\
  \> $\Re(u)$ \> : real part of the complex number $u$ \\
  \> $\Im(u)$ \> : imaginary part of the complex number $u$ \\
%  \> $|u|$ \> : magnitude of the complex number $u$ \\
%  \> $\bar{u}$ \> : complex conjugate of the complex number $u$ \\
%  \> $0_{1\times n}$ \> : zero matrix with dimension $1 \times n$ \\
  \> $\mathcal{D}(.)$ \> : domain of an operator \\
  \> $\mathcal{C}, \mathcal{L}_2$ \> : the space of continuous and square integrable complex functions,\\
  \> \> $\ \ $ i.e., $\mathcal{L}_2([-\tau_{\max},0],\C^n):=\{f:[-\tau_{\max},0] \rightarrow \C^n : \int_{-\tau_{\max}}^0 |f(t)|^2 dt <\infty\}$ \\
%  \> $\|F\|_\infty$ \> : $\mathcal{L}_\infty$ norm of the transfer function $F(j\omega)$ \\
%  \> $\alpha(G)$ \>: the spectral abscissa of $G$, i.e.,\\
%  \> \> $\ \ $$\sup_{\lambda\in\C}\left\{\Re(\lambda): \det(G(\lambda))=0 \right\}$. \\
%  \> $\det(A)$ \> : determinant of the matrix $A$ \\
%  \> $A\otimes B$ \> : Kronecker product of matrices $A$ and $B$ \\

\> $\vec \tau\in\mathbb{R}^m$ \>: short notation for $(\tau_1,\ldots,\tau_m)$ \\
  \> $\mathcal{B}(\vec \tau,\epsilon)$ \> : open ball of radius $\epsilon\in\R^+$ centered at $\vec\tau\in(\R^+)^m$, \\
  \> \> \ \ $\mathcal{B}(\vec \tau,\epsilon):=\{\vec\theta\in(\R)^m : \|\vec\theta-\vec \tau\|<\epsilon\}$ \\
%  $\lambda_i(A):$ & i$^\mathrm{th}$ eigenvalue of $A$,\ $|\lambda_1(\cdot)|\geq|\lambda_2(\cdot)|\geq \cdots$\\
 \> $\left[\begin{array}{c|c}
 A & B \\
 \hline
 C & D \\
 \end{array}\right]
 $ \>: transfer function representation for $T(\lambda)=C(\lambda I - A)^{-1} B+D$ \\
\end{tabbing}

\section{Motivating examples} \label{sec:motex}

With some simple examples we illustrate the generality of the system description (\ref{system}).

\begin{example} \label{elim:connect}
Consider the feedback interconnection of the system
\[
\left\{\begin{array}{lll}
\dot x(t)&=&A x(t)+B_1 u(t)+B_2w(t),\\
 y(t)&=& C x(t)+D_1 u(t),\\
 z(t)&=& F x(t),
\end{array}\right.
\]
and the controller
\[
u(t)=K y(t-\tau).
\]
For $\tau=0$ it is possible to eliminate the output and controller equation, which results in the closed-loop system
\begin{equation}\label{elimination}
\left\{\begin{array}{lll}
\dot x(t)&=& A x(t)+B_1 K (I-D_1 K)^{-1} C x(t)+B_2 w(t), \\
z(t) & =& F x(t).
\end{array}\right.
\end{equation}
This approach is for instance taken in the software package HIFOO~\cite{Burke-hifoo}.
If $\tau\neq 0$, then the elimination is not possible any more. However, if we let $X=[x^T\ u^T y^T]^T$ we can describe the system by the equations
{\small\[
\left\{\begin{array}{l}
\left[\begin{array}{ccc}
I & 0 & 0 \\ 0 &0 & 0 \\ 0& 0&0
\end{array}\right]\dot X(t)=
\left[\begin{array}{ccc}
A & B_1 & 0 \\
C & D_1 &-I\\
0& I & 0
\end{array}\right] X(t)-
\left[\begin{array}{ccc}
0 & 0 & 0 \\
0 & 0 & 0\\
0& 0 & K
\end{array}\right] X(t-\tau) +
\left[\begin{array}{c}B_2\\0\\0
\end{array}\right] w(t),
\\
z(t)=\left[\begin{array}{cc c} F & 0&0  \end{array}\right]X(t),
 \end{array}\right.
 \]}
 which are of the form (\ref{system}).  Furthermore, the dependence of the matrices of the closed-loop system on the controller parameters, $K$, is still linear, unlike in (\ref{elimination}).
\end{example}
\begin{example} \label{elim:feedthru}
The presence of a direct feedthrough term from $w$ to $z$, as in
\begin{equation}\label{ex2}
\left\{\begin{array}{lll}
\dot x(t)&=& Ax(t)+A_1 x(t-\tau)+B w(t),\\
z(t)&=&F x(t)+D_2 w(t),
\end{array}\right.
\end{equation}
can be avoided by introducing a slack variable. If we let $X=[x^T\ \gamma_w^T]^T$, where $\gamma_w$ is the slack variable, we can bring (\ref{ex2}) in the form (\ref{system}):
\[
\left\{\begin{array}{l}
\left[\begin{array}{cc}
I & 0 \\ 0 &0
\end{array}\right]\dot X(t)=
\left[\begin{array}{cc}
A & 0 \\ 0 & -I
\end{array}\right] X(t) +
\left[\begin{array}{cc}
 A_1 & 0 \\ 0 & 0
\end{array}\right] X(t-\tau) +
\left[ \begin{array}{l} B\\ I
\end{array}\right] w(t),
\\
z(t)=[F\ D_2]\ X(t).
\end{array}\right.
\]
\end{example}
\begin{example} \label{elim:inputdelay}
The system
\[
\left\{\begin{array}{lll}
\dot x(t) &=& A x(t)+B_1 w(t)+B_2 w(t-\tau),\\
z(t)&=& C x(t),
\end{array}\right.
\]
can also be brought in the standard form (\ref{system}) by a slack variable. Letting $X=[x^T \gamma_w^T]^T$ we can express
\[
\left\{\begin{array}{lll}
\dot X(t) &=&
\left[\begin{array}{cc}
A & B_1 \\
0 & -I
\end{array}\right]
X(t)
+
\left[\begin{array}{cc}
0 & B_2 \\
0 & 0
\end{array}\right]
X(t-\tau)
+
\left[\begin{array}{c}
0 \\ I
\end{array}\right] w(t),
\\
z(t)&=& [C\ \  0]\ X(t).
\end{array}\right.
\]
In a similar way one can deal with delays in the output $z$.
\end{example}

Using the techniques illustrated with the above examples a broad class of interconnected
systems with delays can be brought in the form (\ref{system}), where the external
inputs $w$ and outputs $z$ stem from the performance specifications expressed in terms of
appropriately defined transfer functions. As a more realistic illustration, the feedback interconnection of any retarded type time-delay system $G$ and controller $K$ with the following state-space representations,

\begin{equation} \label{plant}
G:\left\{\begin{array}{l}
\dot x_G(t)=
 \sum_{i=0}^{m_a}        A^i   x_G(t-\tau_i^a)
+\sum_{i=0}^{m_{b_1}}    B_1^i   w(t-\tau_i^{b_1})
+\sum_{i=0}^{m_{b_2}}    B_2^i   u(t-\tau_i^{b_2}) \\

z(t)=
 \sum_{i=0}^{m_{c_1}}    C_1^i    x_G(t-\tau_i^{c_1})
+\sum_{i=0}^{m_{d_{11}}} D_{11}^i   w(t-\tau_i^{d_{11}})
+\sum_{i=0}^{m_{d_{12}}} D_{12}^i   u(t-\tau_i^{d_{12}}) \\

y(t)=
 \sum_{i=0}^{m_{c_2}}    C_2^i    x_G(t-\tau_i^{c_2})
+\sum_{i=0}^{m_{d_{21}}} D_{21}^i   w(t-\tau_i^{d_{21}})
+\sum_{i=0}^{m_{d_{22}}} D_{22}^i   u(t-\tau_i^{d_{22}}),
\end{array}\right.
\end{equation}

\begin{equation} \label{controller}
K:\left\{\begin{array}{l}
\dot x_K(t)=
 \sum_{i=0}^{m_{a_k}} A_K^i x_K(t-\tau_i^{a_k})
+\sum_{i=0}^{m_{b_k}} B_K^i   y(t-\tau_i^{b_k}) \\

u(t)=
 \sum_{i=0}^{m_{c_k}} C_K^i x_K(t-\tau_i^{c_k})
+\sum_{i=0}^{m_{d_k}} D_K^i   u(t-\tau_i^{d_k}), \\
\end{array}\right.
\end{equation}
can be written in the form of (\ref{system}) using similar techniques in the previous examples.

The price to pay for the generality of the framework is the increase of the dimension of the system, $n$, which affects the efficiency of the numerical methods. However, this is a minor problem in most applications because the delay difference equations or algebraic constraints are related to inputs and outputs, and
the number of inputs and outputs is usually much smaller than the number of state variables.

\section{Preliminaries} \label{sec:prelim}

\subsection*{Assumptions}
Let $\mathrm{rank}(E)=n-\nu$, with $\nu\leq n$, and let the columns
of matrix $U\in\RR^{n\times \nu}$, respectively $V\in\RR^{n\times \nu}$, be a (minimal) basis for
the left, respectively right nullspace, that is,
\[
U^T E=0,\ \ E V=0.
\]
Throughout the paper we make the following assumption.
\begin{assumption} \label{assumption}
The matrix $U^T A_0 V$ is nonsingular.
\end{assumption}

\medskip

In order to motivate Assumption \ref{assumption}, we note that the equations (\ref{system}) can be separated into coupled
delay differential and delay  difference equations. When we
define
\[
\mathbf{U}= \left[{U^{\perp}}\  U\right],\ \ \mathbf{V}=\left[V^{\perp}\ V\right],\
\]
a pre-multiplication of (\ref{system}) with $\UU^T$ and the substitution
\[
x=\VV\ [x_1^T\ x_2^T]^T,
\]
with $x_1(t)\in\RR^{n-\nu}$ and $x_2(t)\in\RR^\nu$, yield the coupled equations
\begin{equation}\label{coupled}
\left\{\begin{array}{ccl}
%
%\begin{eqnarray}
E^\ee \dot x_1(t)&=& \sum_{i=0}^m A_i^\ee x_1(t-\tau_i) +\sum_{i=0}^m A_i^\et x_2(t-\tau_i)+B_1 w(t), \\
0&=&A_0^\tt x_2(t)+ \sum_{i=1}^m A_i^\tt x_2(t-\tau_i) \\
&&\ \hspace*{3.1cm}+\sum_{i=0}^m A_i^\te x_1(t-\tau_i)+B_2 w(t), \\
z(t)&= &C_1 x_1(t)+C_2 x_2(t),
\end{array}\right.
\end{equation}
%
%
%\end{eqnarray}
where
\[
\left.
\begin{array}{lll}
A_i^\ee= {U^{\perp}}^T A_i V^{\perp}, & A_i^\et= {U^{\perp}}^T A_i V,&\\
A_i^\te= {U}^T A_i V^{\perp}, & A_i^\tt= {U}^T A_i V, & i=0,\ldots,m,
\end{array}
\right.
\]
and
\[
E^\ee= {U^{\perp}}^T E V^{\perp},\ \   B_1={U^{\perp}}^T B,\ \ B_2= U^T B,\ \ C_1=C V^{\perp},\ \ C_2=C V.
\]
Matrix $E^\ee$ in (\ref{coupled}) is invertible, following from
\[
n-\nu=\mathrm{rank}(E)=\mathrm{rank}(\mathbf{U}^T E\mathbf{V})=\mathrm{rank}(E^{(11)}).
\]
In addition, matrix $A_0^{(22)}$ is invertible, following from
Assumption~\ref{assumption}.

The equations (\ref{coupled}), with $w\equiv 0$, are semi-explicit delay differential algebraic equations of index~1, because  delay differential equations are obtained by differentiating the second equation. This precludes the occurrence of impulsive solutions~\cite{fridman}. Moreover, the invertibility of $A_0^{(22)}$
 prevents that the equations are of \emph{advanced} type and, hence, non-causal. This further motivates
why Assumption~\ref{assumption} is natural in the delay case considered, although it restricts the index to
one (for a general treatment in the delay free case, see for instance~\cite{stykel} and the references therein).
\medskip

We further make the following assumption.
\begin{assumption} \label{assumption_sstab}
The zero solution of system (\ref{system}), with $w\equiv0$, is strongly exponentially stable.
\end{assumption}

Strong exponential stability refers to the fact that the asymptotic stability of the null solution is robust against small delay perturbations \cite{have:02,Michiels:2007:MULTIVARIATE}. Due to modeling errors and uncertainty, the delays in the  model are usually not exact and this type of stability is required in practice. The stability of the closed-loop system (\ref{system}) is a necessary assumption for $\Hi$ norm optimization since this norm is finite for stable systems only. We assume that parameters of a  controller are available such that the closed-loop system of the form (\ref{system}) is strongly exponentially stable. These parameters can, for instance, be found by minimizing the spectral abscissa of the closed-loop system (\ref{system}) using a non-smooth, non-convex optimization method \cite{suatHIFOO}. The overall $\Hi$ optimization of the closed-loop system (\ref{system}) can then be performed in two steps. First a fixed-order or fixed-structure controller strongly stabilizing the closed-loop system (\ref{system}) is designed. Next the controller parameters are tuned to minimize the $\Hi$ norm of the closed-loop system (\ref{system}) starting from the initial controller obtained in the first step. In this article we focus on the  $\Hi$ computation and optimization.

\subsection*{Transfer functions}

From (\ref{coupled}) we can write the transfer function of the system (\ref{system}) as
\begin{eqnarray}
\label{T} T(\lambda)&:=&C\left(\lambda E-A_0-\sum_{i=1}^m A_i e^{-\lambda \tau_i}\right)^{-1}B, \\
&=&[C_1 \ \ C_2]\left[\begin{array}{rr}\lambda E^\ee-A_{11}(\lambda) & -A_{12}(\lambda)\\
-A_{21}(\lambda) & -A_{22}(\lambda)
 \end{array}\right]^{-1} \left[\begin{array}{c}B_1\\ B_2 \end{array}\right], \label{transferblock}
\end{eqnarray}
with
\[
A_{kl}(\lambda)=\sum_{i=0}^m A_i^{(kl)} e^{-\lambda\tau_i},\ \ k,l\in\{1,2\}.
\]

The {\it asymptotic} transfer function of the system (\ref{system}) is defined as
\begin{eqnarray}
\label{Ta} T_a(\lambda) &:=&-C V \left(U^T A_0 V +\sum_{i=1}^m U^T A_i V e^{-\lambda\tau_i} \right)^{-1} U^TB\\
\nonumber &=&-C_2 A_{22}(\lambda)^{-1} B_2 .
\end{eqnarray}

The terminology stems from the fact that the transfer function $T$ and the asymptotic transfer function $T_a$ converge to each other for high frequencies. This is precisely stated in the following Proposition.
\begin{proposition}\label{propconverge}
$\forall\gamma >0$, $\exists\Omega>0$: $\sigma_{1}\left(T(j\w)-T_a(j\w)\right)<\gamma,\ \forall\w>\Omega$.
\end{proposition}
\begin{proof}
The assertion follows from the explicit expression for the inverse of the two-by-two block matrix in (\ref{transferblock}), combined with the property that
 \begin{equation}\label{normA22}
\sup_{\Re(\lambda)\geq 0}\left \| \left( A_{22}(\lambda) \right)^{-1} \right\|_2
 \end{equation}
 is finite. The latter is due to Assumption~\ref{assumption_sstab}.
\end{proof}

The $\Hi$ norm of the transfer function $T$ of the \emph{stable} system (\ref{system}), is defined as
\[
\|T(j\w)\|_\infty:=\sup_{\w\in \R} \sigma_{1} \left( T(j\w) \right).
\]
Similarly we can define $\Hi$ norm of $T_a$.
\smallskip

\section{The strong H-infinity norm of time-delay systems} \label{sec:shinf}

We now analyze continuity properties of the $\Hi$ norm of the transfer function $T$ with respect to the delay parameters. The function
\begin{equation}\label{defTdel}
\vec\tau\in(\RR_{0}^+)^{m}\mapsto \|T(j\w,\vec\tau)\|_\infty
\end{equation}
is, in general, not continuous, which is inherited from the behavior of the asymptotic transfer function, $T_a$, more precisely the function
\begin{equation}\label{defTadel}
\vec\tau\in(\RR_{0}^+)^{m}\mapsto \|T_a(j\w,\vec\tau)\|_\infty.
\end{equation}
We start with a motivating example

\begin{example} \label{ex:TandTa}
Let the transfer function $T$ be defined as
\begin{equation} \label{Tex}
T(\lambda,\vec\tau)=\frac{\lambda+2.1}{(\lambda+0.1)(1-0.25e^{-\lambda\tau_1}+0.5e^{-\lambda\tau_2})+1}
\end{equation} where $(\tau_1,\tau_2)=(1,2)$. The transfer function $T$ is stable, its $\Hi$ norm is $2.5788$, achieved at $\w=1.6555$ and the maximum singular value plot is given in Figure \ref{fig:svd1}. The high frequency behavior is described by the asymptotic transfer function
\begin{equation} \label{Taex}
T_a(\lambda,\vec\tau)=\frac{1}{(1-0.25e^{-\lambda\tau_1}+0.5e^{-\lambda\tau_2})},
\end{equation}
whose $\Hi$ norm is equal to $2.0320$, which is less than $\|T(j\w,\vec\tau)\|_\infty$. However, when the first time delay is perturbed to $\tau_1=0.99$, the $\Hi$ norm of the transfer function $T$ is $3.9993$, reached at $\w=158.6569$, see Figure \ref{fig:svd099}. The $\Hi$ norm of $T$ is quite different from that for $(\tau_1,\tau_2)=(1,2)$. A closer look at the maximum singular value plot of the asymptotic transfer function $T_a$ in Figure \ref{fig:svd099Ta} shows that the sensitivity is due to the transfer function $T_a$.

\begin{figure}[!h]
    \begin{minipage}[t]{0.45\textwidth}
        \vspace{0pt}
        \includegraphics[width=\linewidth]{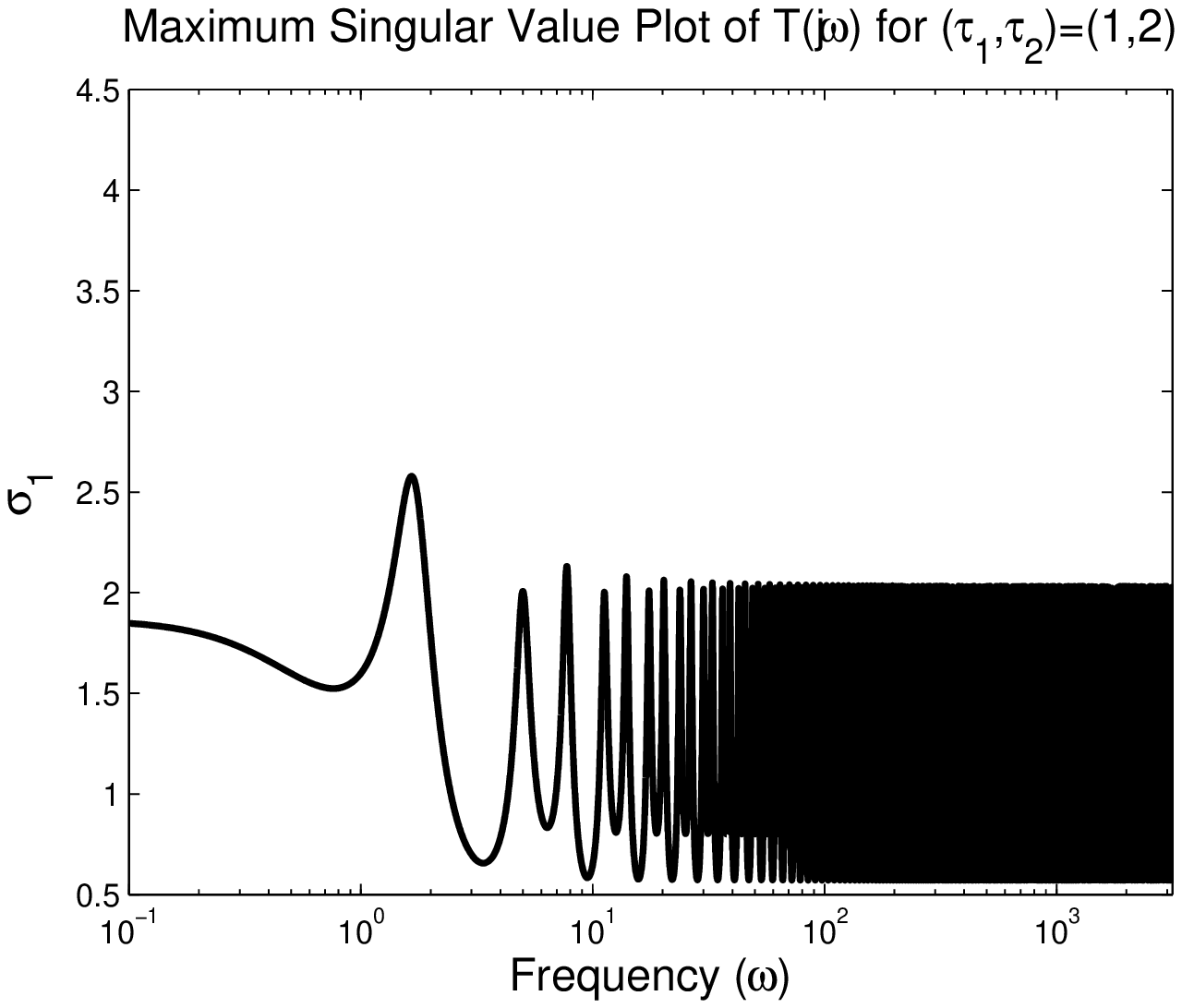}
        \caption{\label{fig:svd1} The maximum singular value plot of $T(j\w,\vec\tau)$ for $(\tau_1,\tau_2)=(1,2)$ as a function of $\omega$.}
    \end{minipage}
\hfill
    \begin{minipage}[t]{0.45\textwidth}
        \vspace{0pt}\raggedright
        \includegraphics[width=\linewidth]{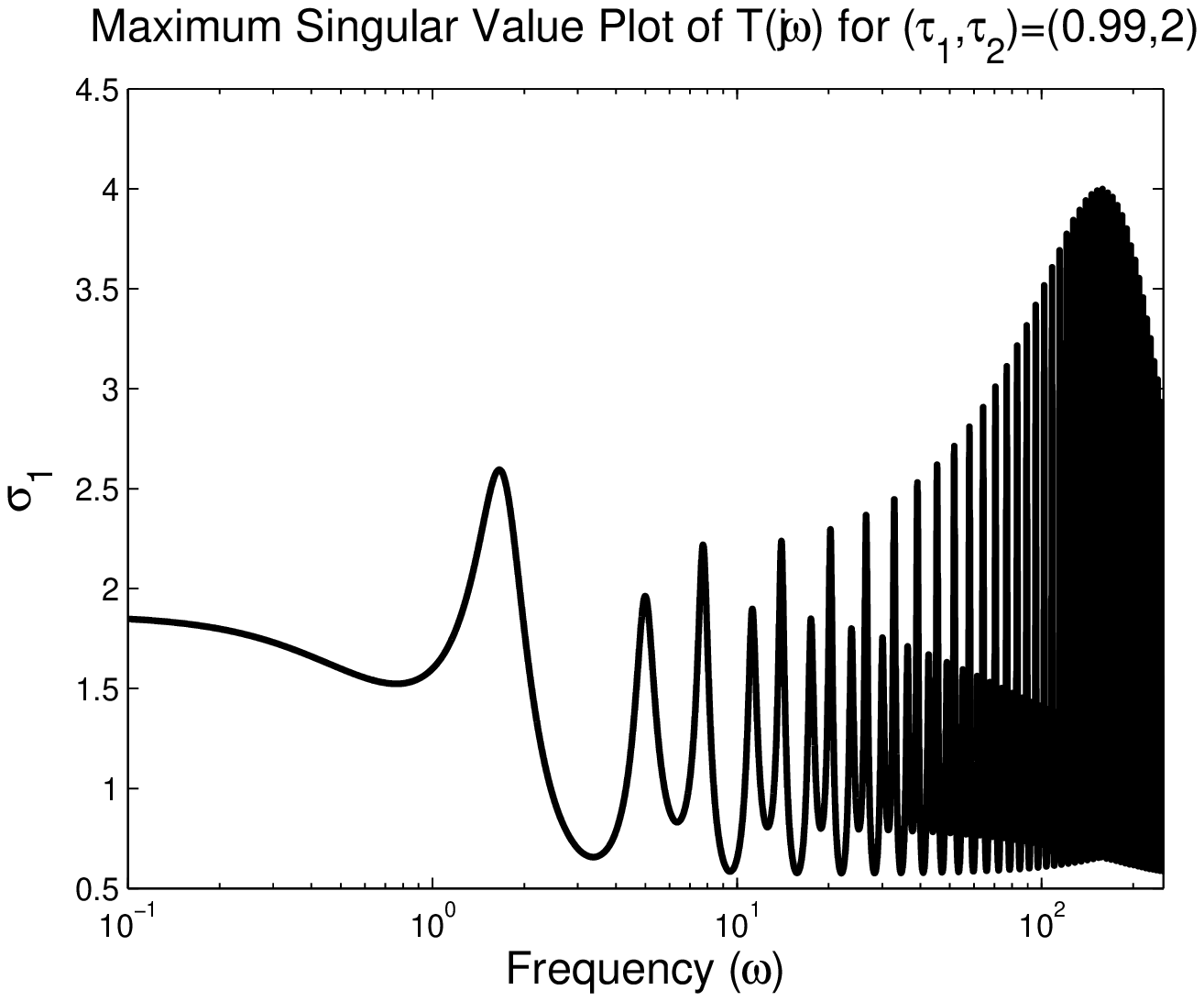}
        \caption{\label{fig:svd099} The maximum singular value plot of $T(j\w,\vec\tau)$ for $(\tau_1,\tau_2)=(0.99,2)$ as a function of $\omega$.}
   \end{minipage}
\end{figure}

Even if the first delay is perturbed slightly to $\tau_1=0.999$, the problem is not resolved, indicating that the functions (\ref{defTdel}) and (\ref{defTadel}) are discontinuous at $(\tau_1,\tau_2)=(1,2)$. The $\Hi$ norm of the transfer function $T$  for $(\tau_1,\tau_2)=(0.999,2)$ is given by $3.9944$, and the peak value is reached at  $\w=1515.8091$. The corresponding asymptotic transfer function $T_a$ is shown in Figure \ref{fig:svd0999Ta}. When the delay perturbation tends to zero, the frequency where the maximum in the singular value plot of the asymptotic transfer function $T_a$ is achieved moves towards infinity.
\end{example}

\begin{figure}[!h]
    \begin{minipage}[t]{0.45\textwidth}
        \vspace{0pt}\raggedright
        \includegraphics[width=\linewidth]{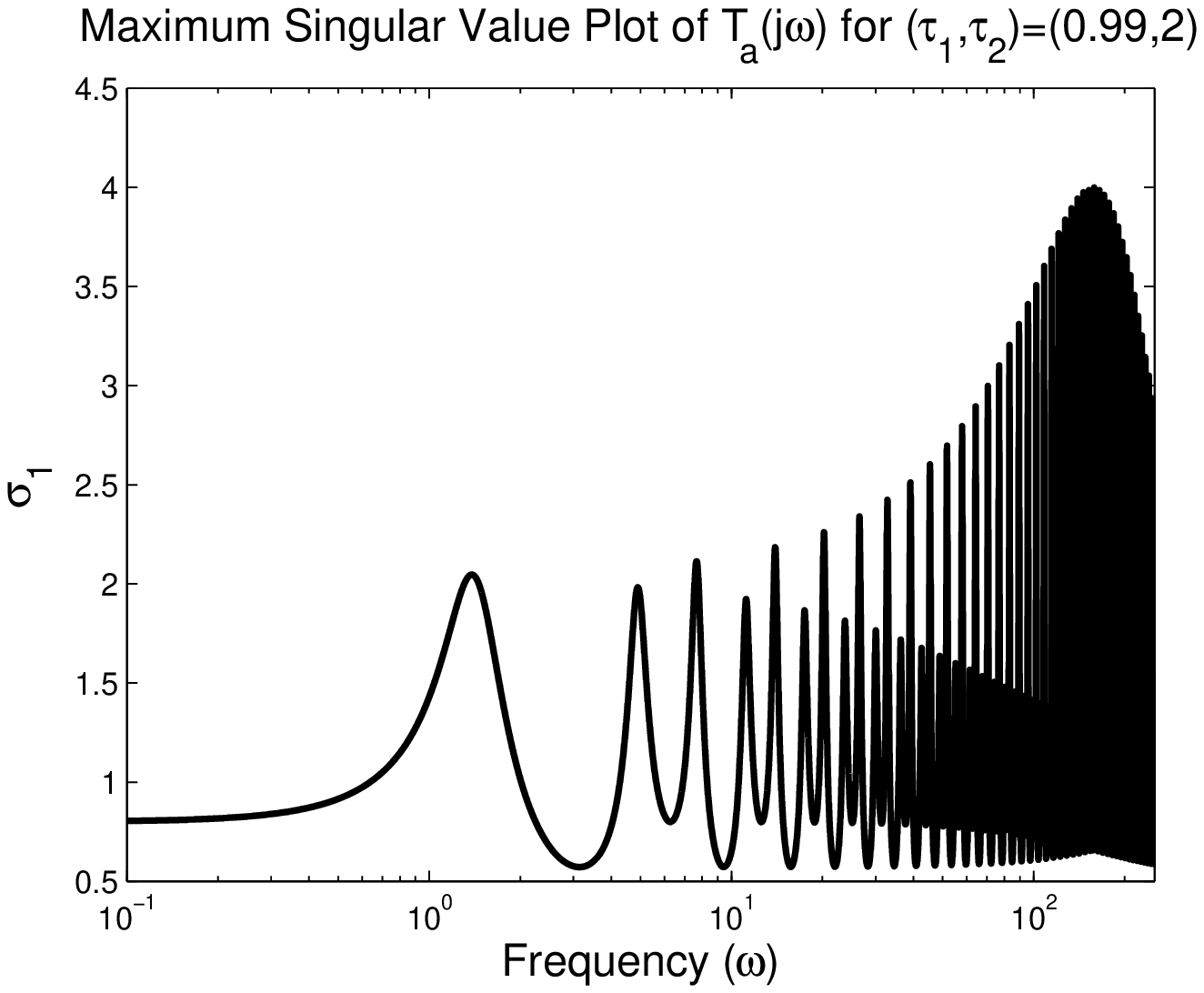}
        \caption{\label{fig:svd099Ta} The maximum singular value plot of $T_a(j\w,\vec\tau)$ for $(\tau_1,\tau_2)=(0.99,2)$ as a function of $\omega$.}
   \end{minipage}
\hfill
    \begin{minipage}[t]{0.45\textwidth}
        \vspace{0pt}
        \includegraphics[width=\linewidth]{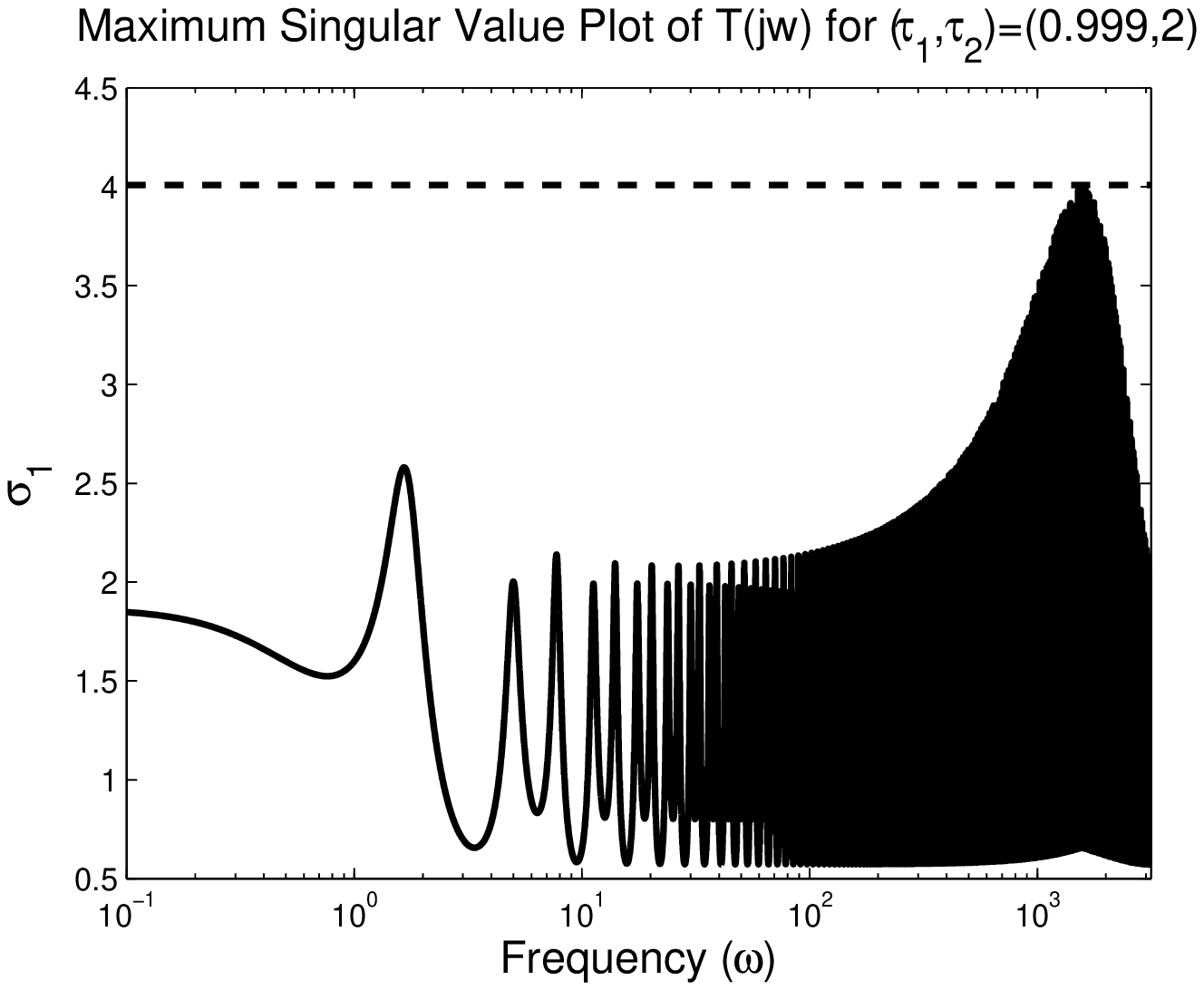}
        \caption{\label{fig:svd0999Ta} The maximum singular value plot of $T_a(j\w,\vec\tau)$ for $(\tau_1,\tau_2)=(0.999,2)$ as a function of $\omega$.}
        \end{minipage}
\end{figure}

The above example illustrates that the $\Hi$ norm of the transfer function $T$ may be sensitive to \emph{infinitesimal} delay changes. On the other hand, for any $\omega_{\max}>0$, the function
\[
\vec\tau\mapsto \max_{[0,\ \omega_{\max}]} \sigma_1(T(j w,\vec\tau)),
\]
where the maximum is taken over a compact set, is continuous, because a discontinuity would be in contradiction with the continuity of the maximum singular value function of a matrix. Hence, the sensitivity of the $\Hi$ norm is related to the behavior of the transfer function at high frequencies and, hence, the asymptotic transfer function $T_a$. Accordingly we start by studying the properties of the function (\ref{defTadel}).

Since small modeling errors and uncertainty are inevitable in a practical design, we wish to characterize the smallest upper bound for the $\Hi$ norm of the asymptotic transfer function $T_a$ which is {\it insensitive} to small delay changes.

\begin{definition} \label{def:shinfTa}
For $\vec \tau\in(\RR_{0}^+)^{m}$, let the strong $\mathcal{H}_{\infty}$ norm of $T_a$, $\interleave {T_a}(j\w,\vec \tau)\interleave_\infty$, be defined as
\[
\interleave T_a(j\w,\vec \tau)\interleave_\infty:=\lim_{\epsilon\rightarrow 0+}
\sup \{\|T_a(j\w,\vec \tau_\epsilon)\|_\infty: \vec \tau_\epsilon\in\mathcal{B}(\vec \tau,\epsilon) \cap (\RR^+)^{m} \} ,
\]
\end {definition}

Several properties of this upper bound on $\|T_a(j\w,\vec \tau)\|_\infty$ are listed below. Recall that  Assumption~\ref{assumption_sstab} is taken.
\begin{proposition}\label{prop:Tasinfprop}
The following assertions hold:
\begin{enumerate}
\item for every $\vec\tau\in(\RR_0^+)^m$, we have
\begin{equation} \label{Tasweep}
\interleave T_a(j\w,\vec \tau)\interleave_\infty=\max_{\vec\theta\in [0,\
2\pi]^m} \sigma_{1} \left( \mathbb{T}_a(\vec \theta) \right),
\end{equation}
where
\begin{equation} \label{Ta_theta}
\mathbb{T}_a(\vec \theta)=-C V \left(U^T A_0 V +\sum_{i=1}^m U^T A_i V e^{-j\theta_i} \right)^{-1} U^T B;
\end{equation}
\item $\interleave T_a(j\w,\vec \tau)\interleave_\infty \geq \|T_a(j\w,\vec \tau)\|_{\infty}$ for all delays $\vec \tau$;
\item $\interleave T_a(j\w,\vec \tau)\interleave_\infty=\|T_a(j\w,\vec \tau)\|_\infty$ for rationally
independent\footnote{The $m$ components of
$\vec\tau=(\tau_1,\ldots,\tau_m)$ are rationally
independent if and only if $\sum_{k=1}^m z_k \tau_k=0,\
z_k\in\ZZ$ implies $z_k=0,\ \forall k=1,\ldots,m$. For
instance, two delays $\tau_1$ and $\tau_2$ are rationally
independent if their ratio is an irrational number.} $\vec \tau$.
\end{enumerate}
\end{proposition}
\begin{proof}
We always have
\begin{multline}\nonumber
(e^{-j\w\tau_1},\ldots,e^{-j\w\tau_m})\in \{(e^{-j\theta_1},\ldots,e^{-j\theta_m}): \theta_i\in[0,2\pi], \ i=1,\ldots,m  \},
\end{multline}
implying
\begin{equation}\label{arg1}
\|T(j\omega,\vec\tau)\|_{\infty}\leq \max_{\vec\theta\in [0,\
2\pi]^m} \sigma_{1} \left( \mathbb{T}_a(\vec \theta) \right).
\end{equation}
For any $\epsilon>0$ in Definition \ref{def:shinfTa}, there exists $\vec \tau_\epsilon=[ \tau_{\epsilon,1},\ldots,\tau_{\epsilon,m}]$ rationally independent in $\mathcal{B}(\vec \tau,\epsilon)\cap(\RR^+)^m$. By Theorem $2.1$ in \cite{TW-report-286}, given rationally independent time delays $\vec \tau_\epsilon$ and for $\vec \theta=[\theta_1,\ldots,\theta_m]$ arbitrary, there exists a sequence of real numbers $\{\w_n\}_{n\geq1}$ such that
\[
\lim_{n\rightarrow\infty} \max_{1\leq i \leq m} \left|e^{-j\w_n \tau_{\epsilon,i}}-e^{-j\theta_i}\right|=0.
\]
It follows that
\begin{multline}
\nonumber \mathrm{closure}\{(e^{-j\w\tau_{\epsilon,1}},\ldots,e^{-j\w\tau_{\epsilon,m}}): \w\in\R  \}= \\
\{(e^{-j\theta_1},\ldots,e^{-j\theta_m}): \theta_i\in[0,2\pi],  i=1,\ldots,m  \},
\end{multline}
implying
\begin{equation}\label{arg2}
\|T(j\omega,\vec\tau_{\epsilon})\|_{\infty}= \max_{\vec\theta\in [0,\
2\pi]^m} \sigma_{1} \left( \mathbb{T}_a(\vec \theta) \right).
\end{equation}
The assertions follow from (\ref{arg1}) and (\ref{arg2}).
%
%
%The first assertion (\ref{Tasweep}) follows. Since $\interleave T_a(j\w,\vec \tau)\interleave_\infty$ is computed over $\mathcal{B}(\vec \tau,\epsilon)$ including nominal delays $\vec \tau$ and does not depend on time-delay values, the second assertion follows. The third assertion is a direct consequence of the fact that for rationally independent $\vec \tau$, there exists a one-to-one mapping between $\w_0\vec\tau$ and $\vec \theta_0$ for all $\w_0\in\R$ and $\vec \theta_0\in[0,\ 2\pi]^m$. The third assertion follows.
\end{proof}

Formula (\ref{Tasweep}) in Proposition \ref{prop:Tasinfprop} shows that the strong $\Hi$ norm of $T_a$ is independent of the delay values. The formula further leads to a computational scheme based on sweeping on $\vec \theta$ intervals. This approximation can be corrected by solving a set of nonlinear equations. Numerical computation details are presented in Section~\ref{sec:hinfnorm_Ta}.

\medskip

We now come back to the properties of the transfer function (\ref{defTdel}) of the system (\ref{system}). As we have illustrated with Example~\ref{ex:TandTa}, a discontinuity of the function (\ref{defTadel}) may carry over to the function (\ref{defTdel}). Therefore, we define the strong $\Hi$ norm of the transfer function~$T$ in a similar way.
\begin{definition}
 For $\vec \tau\in(\RR_{0}^+)^{m}$,  the strong $\Hi$ norm of $T$, $\interleave {T}(j\w,\vec \tau)\interleave_\infty $, is given by
\[
   \interleave T(j\w,\vec \tau)\interleave_\infty:=\lim_{\epsilon\rightarrow 0+}
\sup \{\|T(j\w,\vec \tau_\epsilon)\|_\infty: \vec \tau_\epsilon\in\mathcal{B}(\vec \tau,\epsilon) \cap (\RR^+)^{m} \}.
\]
\end{definition}

The following main theorem describes the desirable property that, in contrast to the $\Hi$ norm, the strong H-infinity norm \emph{continuously} depends on the delay parameters.
It also presents an explicit expression that lays at the basis of the algorithm to compute the strong $\Hi$ norm of a transfer function, presented in the next section. The proof makes use of the technical results in Section~\ref{sec:Appendix2} of the Appendix.
\begin{theorem}
%\begin{enumerate}
%\item
The strong $\Hi$ norm of the transfer function of the DDAE (\ref{system}) satisfies
\begin{equation} \label{shinfnorm}
    \interleave T(j\w,\vec \tau)\interleave_\infty=\max\left(\|T(j\w,\vec\tau)\|_{\infty}, \interleave T_a(j\w,\vec \tau)\interleave_\infty \right),
\end{equation} where $T$ and $T_a$ are the transfer function (\ref{T}) and the asymptotic transfer function (\ref{Ta}).
%\item
%
%\item
In addition, the function
\begin{equation}\label{shinfnorm2}
    \vec \tau\in(\RR^+_0)^m\mapsto \interleave T(j\w,\vec \tau)\interleave_\infty
\end{equation}
is continuous.
%\end{enumerate}
\end{theorem}

\begin{proof}
Lemma~\ref{lem:fininter} implies that the function (\ref{defTdel}) is continuous at delay values where
\begin{equation}\label{condcase1}
\|T(j\omega,\vec\tau)\|_\infty>\interleave T_a(j\omega,\vec\tau)\interleave_\infty.
\end{equation}
This property, along with the fact that $\interleave T_a(j\omega,\vec\tau)\interleave_\infty$ is independent of $\vec\tau$ (see Proposition~\ref{prop:Tasinfprop}), lead to the assertion (\ref{shinfnorm}) and the continuity of (\ref{shinfnorm2}) under the condition (\ref{condcase1}). In the other case the assertions follow from Lemma~\ref{lem3ap}.
\end{proof}

\smallskip

\begin{example}
We come back to Example~\ref{ex:TandTa}. The $\Hi$ norm of $T$, as defined by (\ref{Tex}), is $2.6422$ and the strong $\Hi$ norm of the corresponding asymptotic transfer function $T_a$ is $4$. From property (\ref{shinfnorm}), we conclude that the strong $\Hi$ norm of $T$ (\ref{Tex}) is $4$.
\end{example}

\begin{remark} In contrast to delay perturbations, the $\Hi$ norm of $T$ is continuous with respect to changes of the system matrices $A_i,\ldots ,A_m$, $B$ and $C$.
\end{remark}

\section{Computation of strong H-infinity norms} \label{sec:comp_shinf}

%By Theorem~\ref{shinfnorm}, the strong $\Hi$ norm of the delay differential algebraic system (\ref{system}) can be computed in two steps: First the strong $\Hi$ norm of the asymptotic transfer function $\interleave T_a(j\w,\vec \tau)\interleave_\infty$ is computed. Second the initial level set is set to this norm value and the $\Hi$ norm of the system, $\|T(j\w)\|_\infty$ is computed using a finite-dimensional approximation of the system (\ref{system}) and standard level set methods \cite{boydbala2, steinbuch}. The approximate results are corrected using a local corrector on the set of nonlinear equations.

The algorithm for computing the strong $\Hi$ norm of the transfer function of (\ref{system}) is based on property (\ref{shinfnorm}).  Therefore, we first outline in \S\ref{sec:hinfnorm_Ta} the strong $\Hi$ norm computation of the asymptotic transfer function $T_a$, before presenting the algorithm in~\S\ref{sec:shinfnorm_T}.

\subsection{Strong H-infinity norm of the asymptotic transfer function} \label{sec:hinfnorm_Ta}

The computation of $\interleave T_a(j\w,\vec \tau)\interleave_\infty$ is based on expression (\ref{Tasweep}) in Proposition \ref{prop:Tasinfprop}.  We obtain an approximation by restricting $\vec \theta$ in (\ref{Tasweep}) to a grid,
\begin{equation} \label{Taapprox}
\interleave T_a(j\w,\vec \tau)\interleave_\infty\approx\max_{\vec\theta\in\Theta_h} \sigma_{1} \left(\mathbb{T}_a(\vec \theta)\right),
\end{equation} where $\Theta_h$ is a m-dimensional grid over the hypercube $[0,\ 2\pi]^m$ and $\mathbb{T}_a(\vec \theta)$ is defined by (\ref{Ta_theta}).
If a high accuracy  is required, then the approximate results may be corrected by solving the nonlinear equations
\begin{equation} \label{eq:Tacorrection}
\left\{\begin{array}{l}
\left[\begin{array}{cc}
\mathbb{A}_{22}(\vec \theta) & -\xi^{-1}B_2B_2^T\\
\xi^{-1}C_2^TC_2 & -(\mathbb{A}_{22}(\vec \theta))^{*}
\end{array}\right]
\left[\begin{array}{c}
u_a \\
v_a
\end{array}\right]
=0, \\
n(u_a,v_a)=0, \\
\Re ( e^{-j\theta_i}(v_a^* A_i^{(22)} u_a) )=0,\ i=1,\ldots,m,
\end{array}\right.
\end{equation}
where
\begin{equation}\label{defa22}
 \mathbb{A}_{22}(\vec \theta)=-U^T A_0 V -\sum_{i=1}^m U^T A_i V e^{-j\theta_i}
 \end{equation}
and $n(u_a,v_a)=0$ is a normalization constraint. The first equation in (\ref{eq:Tacorrection}) implies that  $\xi$ is a singular value of $\mathbb{T}_a(\vec\theta)$. The last equation of (\ref{eq:Tacorrection}) expresses that the derivatives of the singular value $\xi$ with respect to the elements of $\vec \theta$ are zero.
%Therefore the solution of (\ref{eq:Tacorrection}) finds the peak $\hat{\xi}$ in the singular value plot with respect to time-delays $\vec{\hat{\theta}}$ starting from the approximate values from the prediction step.
%
In our implementation we solve (\ref{eq:Tacorrection}) using the Gauss-Newton method, which exhibits quadratic convergence because the (overdetermined) equations have an exact solution, see Section 10.2 of \cite{gauss-newton}.

In most practical  problems, the number of delays to be considered in $\mathbb{T}_a(\vec\theta)$ and  $\mathbb{A}_{22}(\vec \theta)$ is much smaller than the number of system delays, $m$, because most of the terms in (\ref{defa22}) are zero. This significantly reduces the computational cost of the sweeping in (\ref{Taapprox}). Note that in a control application a nonzero term in (\ref{defa22}) corresponds to a high frequency feedthrough over the control loop. %In such a case the computation (\ref{Taapprox}) by sweeping over the grid $\Theta_h$ is not computationally expensive.
We illustrate this with the following example.

\begin{example}
Consider the time-delay system
\begin{equation}\label{veeltau}
\left\{\begin{array}{lll}
\dot x(t)&=&\sum_{i=2}^{m} M_i x(t-\tau_i)+ B_1 w(t),\\
 z(t)&=& Px(t)+w(t)+N_1w(t-\tau_1).
\end{array}\right.
\end{equation}
When defining $X=[x^T\ \gamma_d^T \gamma_w^T]^T$, where $\gamma_d$ and $\gamma_w$ are slack variables, the system can be described by equations of the form  (\ref{system}) as
{\small\[
\left\{\begin{array}{l}
\underbrace{\left[\begin{array}{ccc}
I & 0 & 0 \\
0 &0 & 0 \\
0& 0&0
\end{array}\right]}_{E}\dot X(t)=
\underbrace{\left[\begin{array}{ccc}
0 & 0 & 0 \\
0 & -I & I \\
0& 0& -I
\end{array}\right]}_{A_0}X(t)+
\underbrace{\left[\begin{array}{ccc}
0 & 0 & 0 \\
0 & 0 & N_1\\
0& 0 & 0
\end{array}\right]}_{A_{1}} X(t-\tau_1)
\\ \hspace{5.8cm} +
\sum_{i=2}^m
\underbrace{\left[\begin{array}{ccc}
M_i & 0 & 0 \\
0 & 0 &0\\
0& 0 & 0
\end{array}\right]}_{A_i} X(t-\tau_i)+
\underbrace{\left[\begin{array}{c}
B_1\\0\\I
\end{array}\right]}_{B} w(t),
\\
z(t)=\underbrace{\left[\begin{array}{ccc}
P & I & 0
\end{array}\right]}_{C}X(t).
 \end{array}\right.
 \]}
The asymptotic transfer function (\ref{Ta}) is given by
\[
T_a(\lambda)=- CV \left(U^TA_0V+\sum_{i=1}^m U^TA_iV e^{-\lambda\tau_i}\right)^{-1} U^TB,
\] where $U=V=\left(
                     \begin{array}{cc}
                       0 & 0 \\
                       I & 0 \\
                       0 & I \\
                     \end{array}
                   \right)$. Since $U^T A_i V=0$ for $i=2,\ldots,m$, $T_a(\lambda)$ reduces to
\begin{eqnarray}
\nonumber T_a(\lambda)&=&- CV \left(U^TA_0V+U^TA_1V e^{-\lambda\tau_1}\right)^{-1} U^TB, \\
\nonumber &=&-\left[
                \begin{array}{cc}
                  I & 0 \\
                \end{array}
              \right]
              \left[
                \begin{array}{cc}
                  -I & -(I+N_1e^{-\lambda\tau_1}) \\
                  0 & -I \\
                \end{array}
              \right]^{-1}
              \left[
                \begin{array}{cc}
                  0 \\
                  I \\
                \end{array}
              \right]\\
\nonumber &=&I+N_1e^{-\lambda \tau_1},
\end{eqnarray}
which readily follows from (\ref{veeltau}).
Although the original system has $m$ delays, the asymptotic transfer function has only one delay $\tau_1$. Accordingly, the grid $\Theta_h$  in the approximation (\ref{Taapprox}) reduces to a grid on the interval $[0,\ 2\pi]$.
\end{example}

\smallskip

In the numerical implementation, we compute the matrix norm of $U^T A_i V$ for $i=1,\ldots,m$ and omit the corresponding time-delays if their norms are less than a tolerance value.

\subsection{Algorithm} \label{sec:shinfnorm_T}
%The algorithm for the strong $\Hi$ norm computation of the delay differential algebraic system (\ref{system}) is a predictor-corrector method similar to the algorithm for the $\Hi$ norm computation for delay differential system of retarded type in \cite{wimsimax}.
%

%
From (\ref{shinfnorm}) the following implication can be derived.
\[
\interleave T(j\w,\vec \tau)\interleave_\infty> \interleave T_a(j\w,\vec \tau)\interleave_\infty \Rightarrow \interleave T(j\w,\vec \tau)\interleave_\infty=\|T(j\w,\vec \tau)\|_\infty.
\]
Moreover, we learn from Lemma \ref{lem:fininter} that, given a level
\begin{equation}\label{over}
\xi>\interleave T_a(j\w,\vec \tau)\interleave_\infty,
\end{equation}
there are only \emph{finitely} many frequencies $\omega\geq 0$ for which a singular value of $T(j\omega,\vec\tau)$ is equal to $\xi$. These properties allow an adaptation of the standard level set algorithm for $\Hi$ computations for finite-dimensional systems as described in~\cite{steinbuch}. The differences are two-fold. First one has to restrict to the situation where (\ref{over}) holds. This is possible by a preliminary computation of the strong $\Hi$ norm of $T_a$, as outlined in \S\ref{sec:hinfnorm_Ta}, and setting the initial level such that (\ref{over}) is satisfied. Second, the Hamiltonian eigenvalue problem, from which intersections of singular value curves with level sets are computed, is infinite-dimensional, carrying over from the case of retarded time-delay systems discussed in \cite{wimsimax}. Therefore, a discretization is necessary, which brings us to a predictor-corrector approach.
In the predictor step, an approximation of the strong $\Hi$ norm of $T$ (provided it exceeds $\interleave T_a(j\w,\vec \tau)\interleave_\infty$) is obtained by computing
\[
\|T_N(j\omega)\|_{\infty}
\]
using the level set method presented in \cite{steinbuch}.
Here,
\begin{equation} \label{finite2}
T_N(\lambda):={\bf C_N}(\lambda {\bf E_N}-{\bf A_N})^{-1}{\bf B_N}.
\end{equation}
is the transfer function of the system
\begin{equation}
\left\{\begin{array}{lll}
\mathbf{E}_N\dot{z}(t)&=&\mathbf{A}_Nz(t)+\mathbf{B}_Nu(t),\\
\nonumber y(t)&=&\mathbf{C}_Nz(t),
\end{array} \right.
\end{equation}
obtained by a spectral discretization of (\ref{system}) on a grid of $N$ points, see Section~\ref{sec:findimapp} of the appendix for the derivation.
 %The main difference with respect to \cite{steinbuch} lies in the choice of the first level, which is chosen to be at least the strong $\Hi$ norm of the asymptotic transfer function, $\interleave T_a(j\w,\vec \tau)\interleave_\infty$. According to Lemma \ref{lem:fininter}, given a level set $\xi>\interleave T_a(j\w,\vec \tau)\interleave_\infty$, there are \emph{finitely} many frequencies $\omega$ for which for a singular value of $T(j\omega)$ is equal to $\xi$.
The correction step serves to remove the discretization error on the result. It is based on solving a system of nonlinear equations that characterize extrema in the singular value curves. The initial conditions are generated in  the prediction step, assuring that the algorithm converges to the right peak value.
The overall algorithm for the strong $\Hi$ norm computation is as follows.

\begin{algorithm} \label{alg:hinfnorm}

 {\em
 \noindent
  Input: system data, $N$, grid $\Omega_N$ defined by (\ref{defmesh}), candidate critical frequency $\{\omega_1,\ldots,\omega_l\}$ if available, tolerance tol for the prediction step, $\interleave T_a(j\w,\vec \tau)\interleave_\infty$
 }

\begin{enumerate}
\item \underline{\emph{Prediction step:}}
    \begin{enumerate}
    \item calculate the first level,
    \[
    \xi_l=\max\left(\interleave T_a(j\w,\vec \tau)\interleave_\infty, \sigma_{1}\left(T(j\w_1)\right),\ldots,  \sigma_{1}\left(T(j\w_l)\right)\right)
    \]
    \item repeat until break
        \begin{enumerate}
        \item set $\xi:= \xi_l(1 + 2 \mathrm{tol})$
        \item compute all $\w^{(i)}\in\R$ satisfying $\sigma_k \left(T_N(j\w^{(i)})\right)=\xi$. By \cite[Proposition 12]{genin}, this can be done by computing generalized eigenvalues of the pencil
            \begin{equation}\label{pencil}
                \lambda
                \left[
                    \begin{array}{cc}
                        {\bf E_N}& 0\\
                        0 & {\bf E_N}^T
                    \end{array}
                \right]-
                \left[
                    \begin{array}{cc}
                        {\bf A_N}& \xi^{-1}{\bf B_N B_N}^T\\
                        -\xi^{-1}{\bf C_N}^T{\bf C_N} & -{\bf A_N}^T
                    \end{array}
                \right],
            \end{equation}
            whose imaginary axis eigenvalues are given by $\lambda=j\w^{(i)}$.
        \item \textbf{\em if} no generalized eigenvalues $j\w^{(i)}$ of (\ref{pencil}) exist,
            \textbf{\em then}
              \begin{enumerate}
                  \item[] \textbf{\em if} $\xi_l=\interleave T_a(j\w,\vec \tau)\interleave_\infty$, \textbf{\em then} \\
                      \hspace*{1cm}set
                       $
                       \interleave T(j\w,\vec\tau)\interleave_\infty=
                       \interleave T_a(j\w,\vec\tau)\interleave_\infty
                       $\\
                       \hspace*{1cm}quit
                  \item[] \textbf{\em else} \\
                  \hspace*{1cm} let $\w^{(i)}\in\R$ satisfying $\sigma_k \left(T_N(j\w^{(i)})\right)=\xi_l$,\\
                  \hspace*{1cm} set $\tilde \xi=(\xi+\xi_l)/2$,  $\tilde{\w}^{(i)}=\w^{(i)},\ i=1,2,\ldots$ \\
                  \hspace*{1cm} break, go to the correction step 2.\\
                  \textbf{\em endif}
              \end{enumerate}
              \textbf{\em else}
              \begin{enumerate}
                  \item[] calculate $\mu^{(i)}:=\sqrt{\w^{(i)}\w^{(i+1)}}$, $i=1,2,\ldots$
                  \item[] set
                  \[
                  \xi_l:=\max_i \max \left(\sigma_{1}\left( T_N(j\mu^{(i)})\right),\interleave T_a(j\w,\vec \tau)\interleave_\infty \right).
                  \]
              \end{enumerate}
              \textbf{\em endif}
        \end{enumerate}
    \end{enumerate}
\item \underline{\emph{Correction step:}} \\
%Correct the approximate strong $\Hi$ norm $\tilde{\xi}$ and its frequencies $\tilde{\w}^{(i)}$ using a locally convergent method which is based on solving equations characterizing extreme values in the singular value plot,
%
\begin{enumerate}
\item Solve the nonlinear equations
\begin{equation} \label{eq:Tcorrection}
\left\{\begin{array}{l}H(j\omega,\xi)
\left[\begin{array}{c}
u \\
v
\end{array}\right]
=0, \\
n(u,v)=0, \\
\Im \{v^* (E+\sum_{i=1}^m A_i\tau_ie^{-j\w\tau_i})u \}=0,
\end{array}\right.
\end{equation}
using the Gauss-Newton method, where
\begin{multline}
\nonumber H(j\omega,\xi)=\\
\left[\begin{array}{cc}
j\omega E-A_0-\sum_{i=1}^m A_i e^{-j\omega\tau_i} & -\xi^{-1}BB^T\\
\xi^{-1}C^TC & j\omega E^T+A_0^T+\sum_{i=1}^m A_i^T e^{j\omega\tau_i}
\end{array}\right]
\end{multline}
and $n(u,v)=0$ is a normalizing condition, with the starting values
\[
\omega=\tilde\omega^{(i)},\ \ \xi=\tilde\xi,\
\left[\begin{array}{c}u\\
v\end{array}\right]=
\arg\min{\|H(j\tilde\omega^{(i)},\tilde\xi)\zeta\|}/{\|\zeta\|};
 \]
denote the solutions with $(\hat u^{(i)},\hat
 v^{(i)},\hat\omega^{(i)},\hat\xi^{(i)})$, for $i=1,2,...,$
\item set
$\interleave T(j\omega)\interleave_\infty:=\max_{1\leq i\leq
p}\hat \xi^{(i)}$
\end{enumerate}
%
%The corrected strong $\Hi$ norm $\hat{\xi}$ and its frequency $\hat{\w}$ are returned.
\end{enumerate}
\end{algorithm}

\medskip

\smallskip

The first and the second equation in (\ref{eq:Tcorrection}) describe the presence of a singular value $\xi$ of matrix $T(j\omega,\vec\tau)$. The third equation expresses that the derivative of this singular value with respect to $\omega$ is equal to zero, see~\cite{wimsimax}. Hence, Equations (\ref{eq:Tcorrection}) can be used to correct approximate peak values. Note that the correction step is only performed~if
\[
\interleave T(j\omega,\vec\tau)\interleave_{\infty}>\interleave T_a(j\omega,\vec\tau)\interleave_{\infty}.
\]
 For details on the choice of the number of discretization points, $N$, and the tolerance, tol, we refer to~\cite{wimsimax}.
%
%Note that the correction step is only done if there exists $\w^{(i)}$ in the first iteration. Otherwise the strong $\Hi$ norm $\hat{\xi}$ is equal to $\interleave T_a(j\w,\vec \tau)\interleave_\infty$ and its corresponding delays $\vec{\hat{\theta}}$ are the outputs of the algorithm.
%
%\smallskip
%
%WHY DISCRETIZATION

\medskip

The main ideas behind  Algorithm~\ref{alg:hinfnorm} are clarified with two  examples.
\begin{example} \label{ex:hinfalg}
\begin{figure}[!h]
    \begin{minipage}[t]{0.45\textwidth}
        \vspace{0pt}\raggedright
        \includegraphics[width=\linewidth]{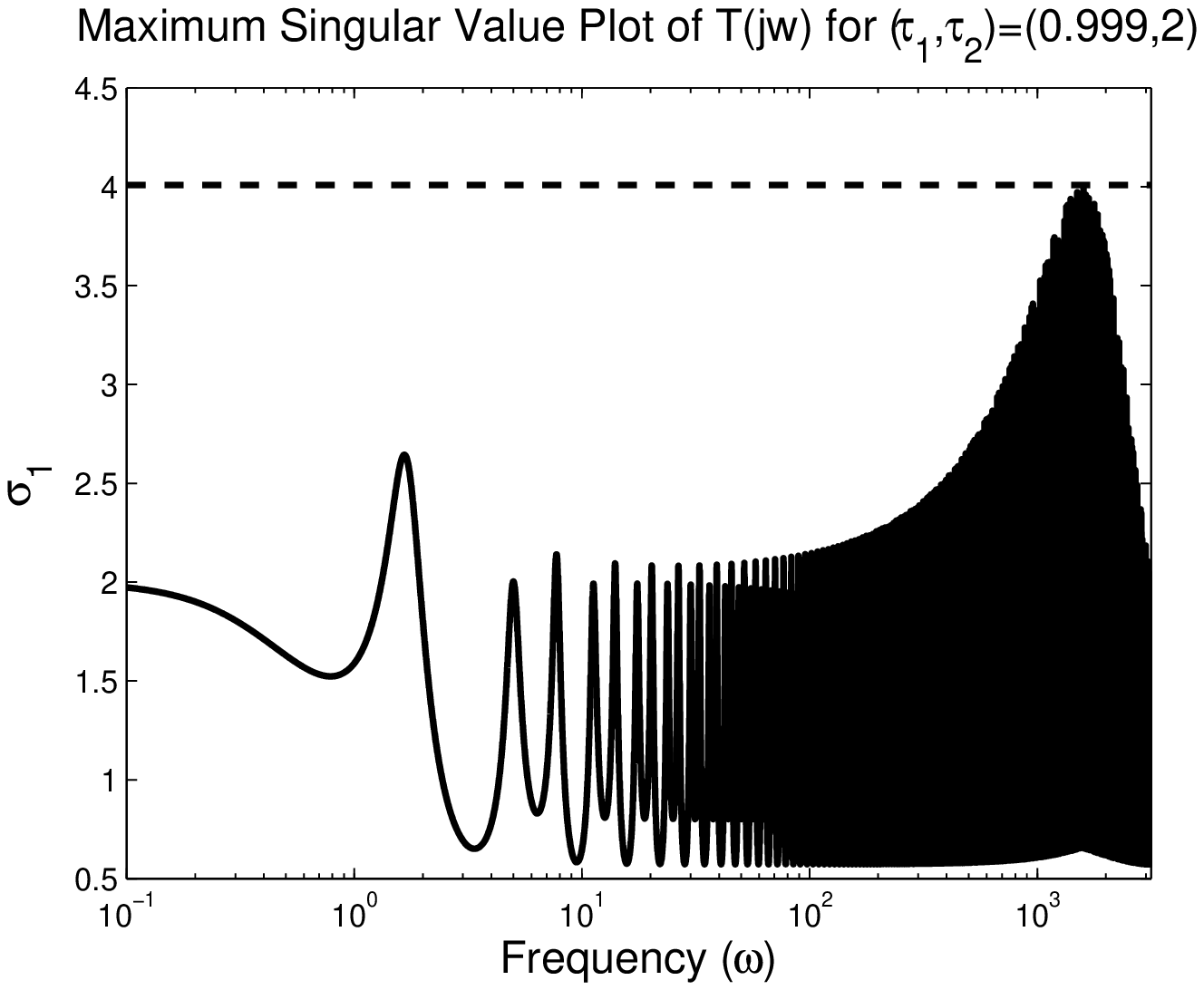}
        \caption{\label{fig:hinfalg_nox} Algorithm~\ref{alg:hinfnorm} for the maximum singular value plot of $T$ (\ref{Tex}): no intersection case.}
   \end{minipage}
\hfill
    \begin{minipage}[t]{0.45\textwidth}
        \vspace{0pt}
        \includegraphics[width=\linewidth]{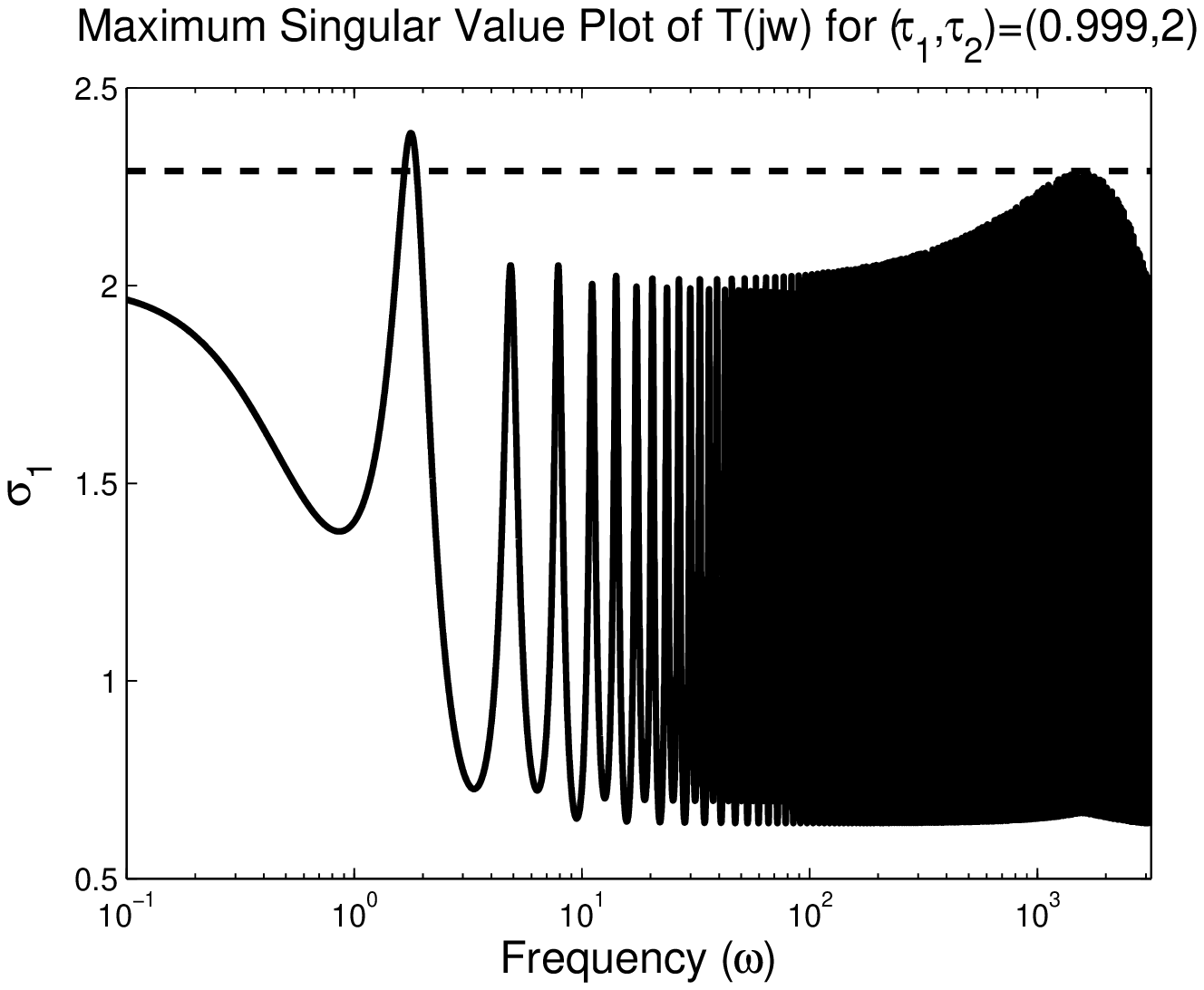}
        \caption{\label{fig:hinfalg_wx} Algorithm~\ref{alg:hinfnorm} for the maximum singular value plot of $T$ (\ref{Tex2}): with intersections case.}
        \end{minipage}
\end{figure}
\begin{figure}
\includegraphics[width=\linewidth]{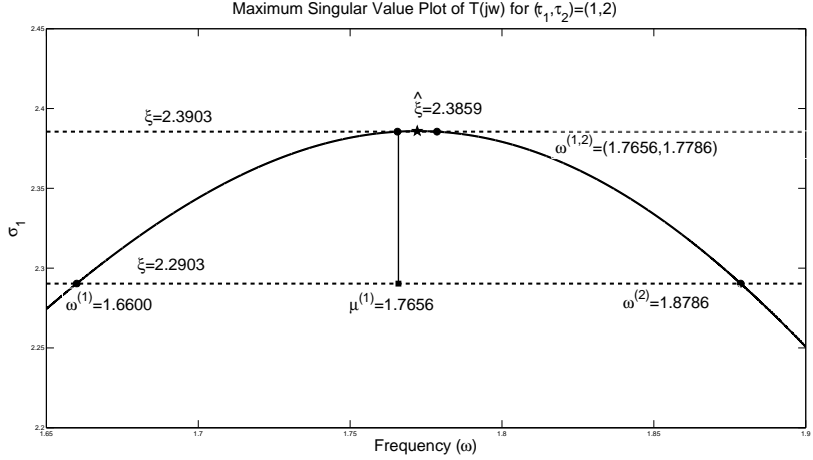}
\caption{\label{fig:hinfalg_wx_zoomed} Steps of Algorithm~\ref{alg:hinfnorm} for the maximum singular value plot of $T$ (\ref{Tex2}): with intersections case (zoomed).}
\end{figure}

We apply Algorithm~\ref{alg:hinfnorm} to the transfer function $T$, as specified by (\ref{Tex}). The strong $\Hi$ norm of its asymptotic transfer function $T_a$ (\ref{Taex}) satisfies $\interleave T_a(j\w,\vec \tau)\interleave_\infty=4$. Therefore, the first level is equal to $\xi_l=4$ in Step $(a)$ of the prediction step of Algorithm~\ref{alg:hinfnorm}, provided that no candidate frequencies are given. In Step $(b)-ii$, there is no intersection for the level $\xi$ as shown in Figure~\ref{fig:hinfalg_nox}. Therefore the strong $\Hi$ norm of $T$ (\ref{Tex}) is equal to $4$ and the correction step is not carried out.
\end{example}

\begin{example}
We consider the transfer function
\begin{equation} \label{Tex2}
T(\lambda,\vec\tau):=\frac{\lambda+2}{\lambda(1-1/16e^{-\lambda\tau_1}+1/2e^{-\lambda\tau_2})+1},
\end{equation}
with $\vec\tau=(1,2)$, and its asymptotic transfer function
\begin{equation} \label{Taex2}
T_a(\lambda,\vec\tau):=\frac{1}{(1-1/16e^{-\lambda\tau_1}+1/2e^{-\lambda\tau_2})}.
\end{equation}
The first level $\xi_l$ in Step $(a)$ is set to the strong $\Hi$ norm of $T_a$ (\ref{Taex2}), $\interleave T_a(j\w,\vec \tau)\interleave_\infty=2.2857$. In Step $(b)-ii$, there are two intersections for the level $\xi=2.2903$, $\w^{(1)}=1.6600$ and $\w^{(2)}=1.8786$, as shown in Figure~\ref{fig:hinfalg_wx}. We can see the details of the next iterations in Figure~\ref{fig:hinfalg_wx_zoomed}. Step $(b)-iii$ calculates the middle frequency $\mu^{(1)}=1.7656$ and $\xi_l=2.3855$ for the next level. In the second iteration, the level is set to $\xi=2.3903$ and the corresponding intersections are $\w^{(1)}=1.7656$ and $\w^{(2)}=1.7786$. Since there is no intersection in the third iteration due to the chosen tolerance in the prediction step, $tol=10^{-3}$, we compute the approximate strong $\Hi$ norm of $T$ (\ref{Tex2}) and the frequencies as $\tilde{\xi}=2.3879$ and $(\tilde{\w}_1,\tilde{\w}_2) =\{1.7657,1.7786\}$. In the correction step, these values are corrected and the strong $\Hi$ norm of $T$ (\ref{Tex2}) and the corresponding frequency  are computed as $\hat{\xi}=2.3859$ and $\hat{\w}=1.7721$.
\end{example}

\section{Fixed-Order H-infinity Controller Design} \label{sec:design}

We consider the equations
\begin{equation}\label{parp}
\left\{\begin{array}{l}
E \dot x(t)=A_0(p) x(t)+\sum_{i=1}^m A_i(p) x(t-\tau_i)+Bw(t),\\
z=C x(t),
\end{array}\right.
\end{equation}
where the system matrices smoothly depend on parameters $p$.
As illustrated in Section~\ref{sec:motex}, a broad class of interconnected systems can be brought into this form, where the parameters $p$ can be interpreted in terms of a parameterization of a controller. For example, in the feedback interconnection of  (\ref{plant}) and (\ref{controller}) they may correspond to the elements of the matrices of the controller $K$. Note that, by fixing some elements of these matrices, additional structure can be imposed on the controller, e.g. a proportional-integrative-derivative (PID) like structure.

 The proposed method for designing fixed-order/ fixed-structure $\Hi$ controllers is based on a direct minimization of the strong $\Hi$ norm of the closed-loop transfer function $T$ from $w$ to $z$ as a function of the parameters $p$.  The overall optimization algorithm requires the evaluation of the objective function and its gradients with respect to the optimization parameters, whenever it is differentiable.

The strong $\Hi$ norm of the transfer function $T$ can be computed by Algorithm~\ref{alg:hinfnorm}. The derivatives of the norm with respect to controller parameters exist whenever there are unique  values $\vec{\hat{\tau}}$ or $\vec {\hat\theta}$ such that
\[
\interleave T(j\w,\vec\tau)\interleave_\infty=\hat\xi=
\left\{\begin{array}{ll}
\sigma_{1}\left(\mathbb{T}_a\left(\vec {\hat \theta}\right)\right),
& \mathrm{if}\ \hat \xi=\interleave T_a\left(\vec {\hat\theta},\vec\tau\right)\interleave_\infty,
\\
\sigma_{1}({T}(j\hat\w)), & \mathrm{if}\ \hat \xi>\interleave T_a\left(\vec {\hat \theta},\vec\tau\right)\interleave_\infty,
\end{array}\right.
\] holds and, in addition, the largest singular value $\hat\xi$ has multiplicity one. We compute the derivative of the strong $\Hi$ norm of $T$ with respect to the  parameter $p_k$ as
\[
\frac{\partial \xi}{\partial p_{k}}=\left\{
\begin{array}{ll}
-2\xi^2\left.\frac{\Re\left(v_a^*\frac{\partial \mathbb{A}_{22}(\vec \theta)}{\partial p_{k}} u_a\right)}
{v_a^*B_2B_2^Tv_a+u_a^*C_2^TC_2u_a}\right|_{(\xi,\vec \theta)=\left(\hat{\xi},\vec{\hat{\theta}}\right)} & \mathrm{if}\ \ \hat{\xi}=\interleave T_a\left(\vec{\hat\theta},\vec \tau\right)\interleave_\infty,\\
-2\xi^2\left.\frac{\Re\left(v^*\frac{\partial A(j\w)}{\partial p_{k}} u\right)}
{v^*BB^Tv+u^*C^TCu}\right|_{(\xi,\w)=(\hat{\xi},\hat{\w})} & \mathrm{if}\ \hat \xi>\interleave T_a\left(\vec{\hat\theta},\vec\tau\right)\interleave_\infty,
\end{array}\right.
\] where given $\xi=\hat \xi$, $u_a,v_a$ and $u,v$ are vectors in (\ref{eq:Tacorrection}) and (\ref{eq:Tcorrection}) for $\vec\theta=\vec{\hat{\theta}}$ and $\w=\hat\w$ respectively. For more details on the computation of derivatives we refer to \cite{thesismarc,bfgbookchapter}.

 The overall design procedure is fully automated and does not require any interaction with the user. The computation cost of the optimization algorithm is dominated by the evaluation of the strong $\Hi$ norm of the closed-loop transfer function $T$ for the parameters $p$ at each iteration. The first main part in this computation is to find the strong $\Hi$ norm of the asymptotic transfer function by computing the maximum singular value of $\mathbb{T}_a$  at $p_a^{m_a-1}$ points spanning the grid $\Theta_h$ in (\ref{Taapprox}),  where $p_a$ is the number of grid points in the interval $[0,\ 2\pi]$ (the default value is $20$ in our implementation)  and $m_a$ is the number of  actual delays appearing in $\mathbb{A}_{22}(\vec \theta)$, see (\ref{defa22}). Note that the number of delays $m_a$  is usually much smaller than the number of system delays, see the arguments at the end of \S\ref{sec:hinfnorm_Ta}. Therefore the computational cost for sweeping is usually not very high. It is even completely skipped if there is no high frequency feedthrough in the control loop (which results in $m_a=0$). The second main part is the computation of the generalized eigenvalues of the pencil (\ref{pencil}) in the prediction step of Algorithm~\ref{alg:hinfnorm}. This computation requires solving a generalized eigenvalue problem with dimensions $2nN$ where the default value for $N$ is $20$ in our implementation. The number of iteration steps of the optimization algorithm heavily depends on the optimization problem under consideration. In most cases, satisfactory results are already obtained in the first phase of the optimization algorithm where the BFGS algorithm is used. For the behavior of BFGS, applied to nonsmooth problems, we refer to \cite{overtonbfgs}.

Recall that  the feedback interconnection of system (\ref{plant}) and controller (\ref{controller}) can be rewritten in the form (\ref{parp}) in such a way that the closed-loop matrices depend affinely on the matrices of the controller. This property improves the performance of the optimization method. Note that in the existing work for systems without delay (see, e.g., \cite{suatHIFOO}) the dependency is in general nonlinear, due to the use of elimination  for handling a non-trivial feedthrough, as illustrated in Example~\ref{elim:connect}.

%Based on the computation of the strong $\Hi$ norm and its derivatives with respect to the controller matrices, we minimize the strong $\Hi$ norm of the closed-loop system (\ref{system}) and design a fixed-order $\Hi$ controller using a non-smooth and non-convex optimization method.

\section{Examples} \label{sec:ex}

In \S\ref{exex1} we illustrate some aspects of the proposed approach on two motivating examples.  In \S\ref{sec:collection} apply the approach to benchmark examples collected from the literature. In \S\ref{sec:bencmarks} we consider  $5$  additional problems.

\subsection{Motivating Examples}\label{exex1}
As a first example we consider  a plant with the state-space representation
\[
\dot{x}(t)=-x(t)-0.5x(t-1)+w(t)+u(t-0.2), \quad z(t)=x(t)+u(t-0.2),\quad y(t)=x(t),
\] and a controller  $u(t)=Ky(t)$. The closed-loop transfer function can be written as
\begin{equation} \label{ex1:T}
T(s)=\frac{1+Ke^{-0.2s}}{s+1-Ke^{-0.2s}+0.5e^{-s}}.
\end{equation} The closed-loop system is internally stable for $-7.9<K<1.5$. As illustrated in Figure~\ref{fig:tds1}, the closed-loop system achieves the minimum strong $\Hi$ norm $\hat{\xi}$ for $K=-0.8813$. The iterations of the optimization method (starting at $K=-7.4$ and shown in circles) converge to the minimum, $0.2137$.

\begin{figure}[t]
    \begin{minipage}[t]{0.45\textwidth}
        \vspace{0pt}\raggedright
        \includegraphics[width=\linewidth]{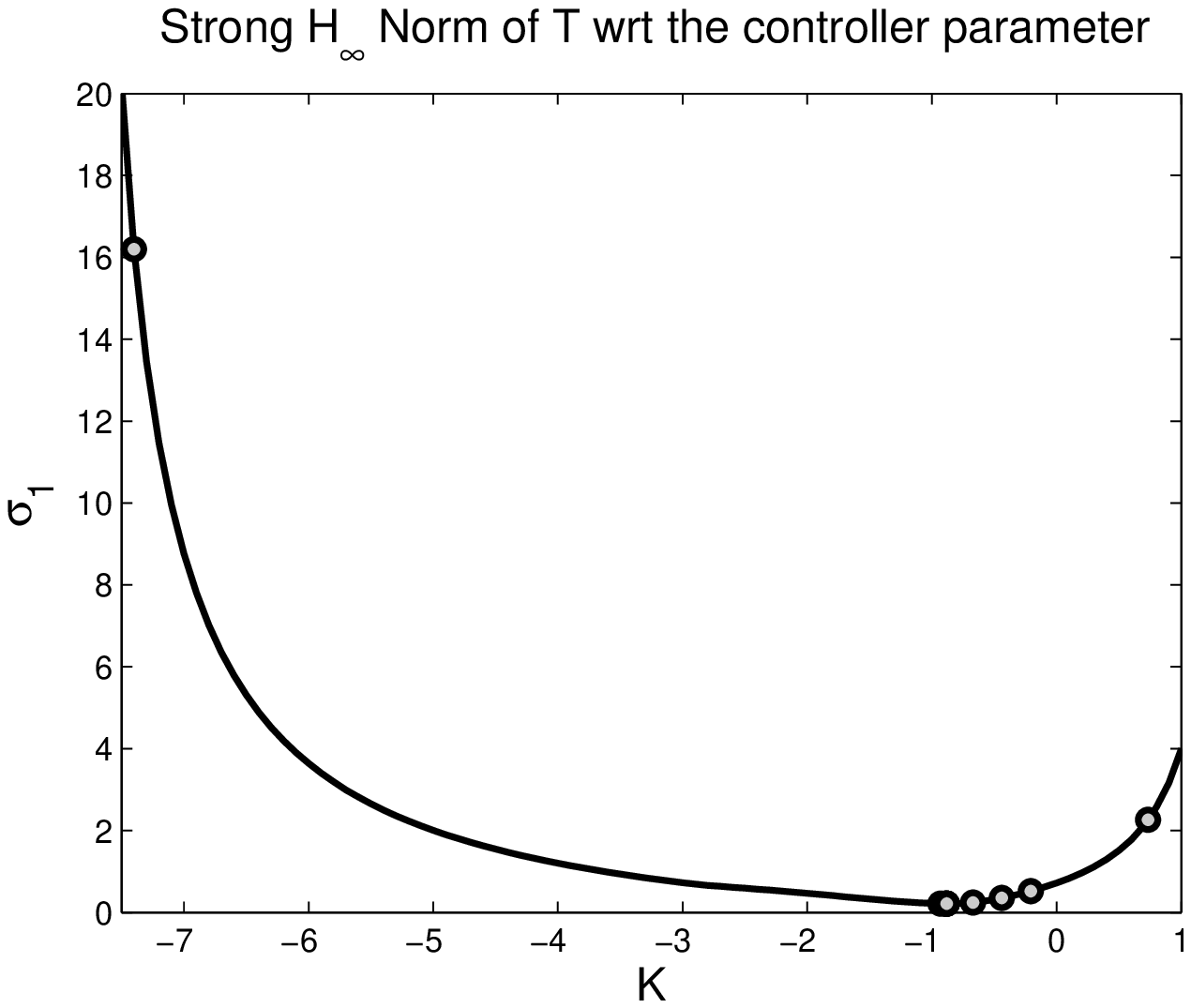}
        \caption{\label{fig:tds1} The strong $\Hi$ norm of $T$ (\ref{ex1:T}) with respect to the controller parameter.}
   \end{minipage}
\hfill
    \begin{minipage}[t]{0.45\textwidth}
        \vspace{0pt}
        \includegraphics[width=\linewidth]{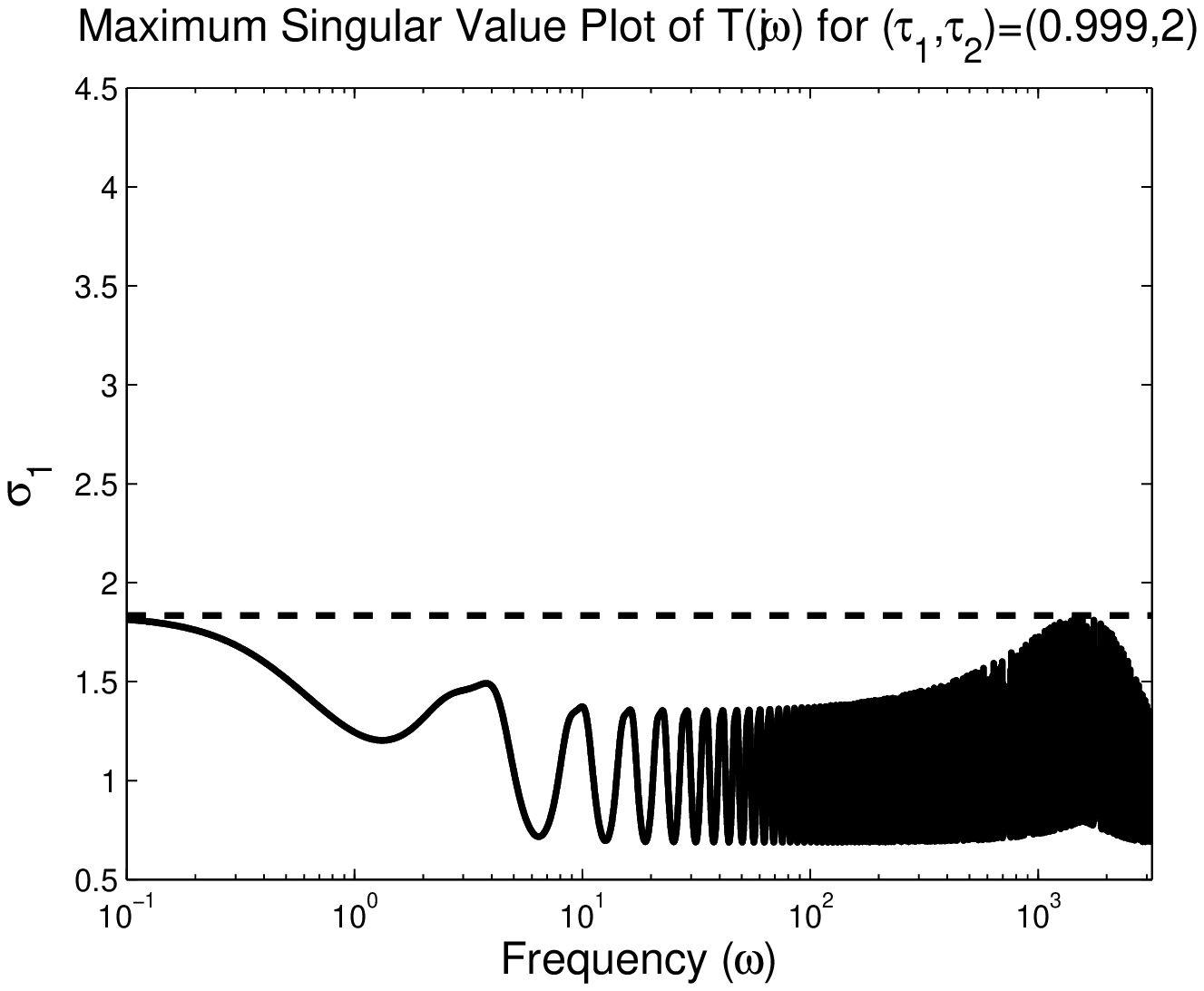}
        \caption{\label{fig:shinfnorm_0992T} The maximum singular value plot of the closed-loop system $T(j\w)$ (\ref{Tex}) for $(\tau_1,\tau_2)=(0.999,2)$ with the controller $K_{opt}$.}
        \end{minipage}
\end{figure}

 The second example concerns the design of a static controller for the case where the standard $\Hi$ norm and the strong $\Hi$ norm of the closed-loop systems are different. The transfer function $T$, as defined by (\ref{Tex}),  can be interpreted as the transfer function of the  closed-loop system formed by the plant
\[
\left\{\begin{array}{rll}
\left[
\begin{array}{cc}
1 & 0 \\
0 & 0
\end{array}
\right] \dot x(t)
&=&
\left[
\begin{array}{cc}
-0.1 & -1 \\
1 & -1
\end{array}
\right] x(t)
+
\left[\begin{array}{c}
0 \\
1
\end{array}\right] u(t)
+ \left[\begin{array}{c}
0 \\
1
\end{array}\right] w(t),\\
 z(t)&=&
 \left[
 \begin{array}{cc}
 2 & -1
 \end{array}
 \right] x(t),\\
 y(t)&=&
 \left[
 \begin{array}{cc}
 0 & 1 \\
 0 & 0
 \end{array}
 \right]
 x(t-\tau_1)
 +
 \left[
 \begin{array}{cc}
 0 & 0 \\
 0 & 1
 \end{array}
 \right]
 x(t-\tau_2),
\end{array}\right.
\]
where $(\tau_1,\tau_2)=(1,2)$ and the controller
\[
u(t)=K y(t),
\] where $K=K_{\textrm{init}}=\left[\begin{array}{cc} 0.25 & -0.5 \end{array}\right]$. In Example~\ref{ex:hinfalg}, we computed the standard $\Hi$ norm $2.5788$ and the strong $\Hi$ norm $4$, as illustrated in Figure~\ref{fig:hinfalg_nox} with slightly perturbed delay values. An optimization of the strong $\Hi$ norm results in
\[
K=K_{opt}=\left[\begin{array}{cc} -0.3533 & -0.1012 \end{array}\right]
\] and the corresponding optimal value is given by $1.8333$. As shown in Figure~\ref{fig:shinfnorm_0992T}, the optimization method pushes the strong $\Hi$ norm of the asymptotic transfer function until it is equal to the standard $\Hi$ norm. Hence, the minimum is characterized by a balance between low and high frequency behavior of the transfer function.

\subsection{A collection of examples from the literature} \label{sec:collection}

We collected benchmark examples for $\Hi$ optimization of time-delay systems from the literature. We considered two types of problems: the $\Hi$ optimization with \emph{state} and \emph{output} feedback controllers. Our results are given in Table~\ref{table:state} and \ref{table:output} respectively.
\begin{table}[!h]
\begin{center}
\begin{tabular}{cccl}
  \hline
  \hline
  Problem & Other Methods & Results & Computed Controller\\
  \hline
  Ex. $4$, \cite{Fridman:2002:DESCRIPTOR} & $1.8822$, \cite{Fridman2001SCL} & $0.1000$ & {\small $\left[-2.3273,\  -9.5004\ 10^3\right]$}\vspace{.5mm}\\
   & $0.2284$, \cite{Fridman2001TAC}  & & \vspace{.5mm}\\
   & $0.1287$, \cite{Fridman:2002:DESCRIPTOR}  & &\vspace{.5mm}\\
  \hline
  Ex. $1$, \cite{Fridman1998SCL} & $0.4215$, \cite{Fridman1998SCL} & $0.4005$ & {\small $\left[-17.8065,\  9.5915\right]$}\vspace{.5mm}\\
  \hline
  Ex. $2$, \cite{fridman} & $21$, \cite{fridman} & $2.9091$ & {\small $\left[-1.1151\ 10^3,\ -1.6189\ 10^4 \right]$}\vspace{.5mm}\\
  \hline
  \hline
\end{tabular}
\end{center}
\caption{The achieved $\Hi$ performances by state-feedback controllers} \label{table:state}
\end{table}

The strong $\Hi$ norm of the closed-loop system in Example $4$ of~\cite{Fridman:2002:DESCRIPTOR} with respect to controller parameters is visualized in Figure~\ref{fig:tds2}. The closed-loop system is stable for sufficiently large negative $K_2$ and its norm converges to $0.1$ as $K_2\rightarrow -\infty$, for any value of $K_1$. This is confirmed by the property that all suggested controllers in the literature have large negative $K_2$ values. The starting point $K=\left[0,\  -5\right]$ and other points are shown as a red dot and gray dots respectively in Figure~\ref{fig:tds2}. The iterations of the optimization method get closer to the value $0.1$ for large negative values of $K_2$ without a significant change in the $K_1$ parameter. Note that there is no finite minimum and the algorithm stops when the number of iterations exceeds the maximum number of iterations in the algorithm. By analyzing similar plots as in Figure~\ref{fig:tds2}, we confirmed that the optimization method reaches close to optimal values for the other two examples in Table~\ref{table:state}. Note that Example $2$ concerns a DDAE while the others concern retarded time-delay systems.

The strong $\Hi$ norm of the closed-loop system in Example $4$ of~\cite{Fridman:2002:DESCRIPTOR} with respect to controller parameters is given in Figure~\ref{fig:tds2}. The closed-loop system is stable for sufficiently large negative $K_2$ and its norm converges to $0.1$, for any value of $K_1$. This is confirmed by the property that all suggested controllers in the literature have large negative $K_2$ values. The starting point $K=\left[0,\  -5\right]$ and other points are shown as a red dot and gray dots respectively in Figure~\ref{fig:tds2}. The iterations of the optimization method get closer to the value $0.1$ for large negative values of $K_2$ without a significant change in the $K_1$ parameter. By analyzing similar plots as in Figure~\ref{fig:tds2}, we confirmed that the optimization method reaches close to optimal values for the other two examples in Table~\ref{table:state}. Note that Example $2$ concerns a DDAE while the others concern retarded time-delay systems.

\begin{figure}[h]
\centering
\vspace{0pt}
\includegraphics[width=0.7\linewidth]{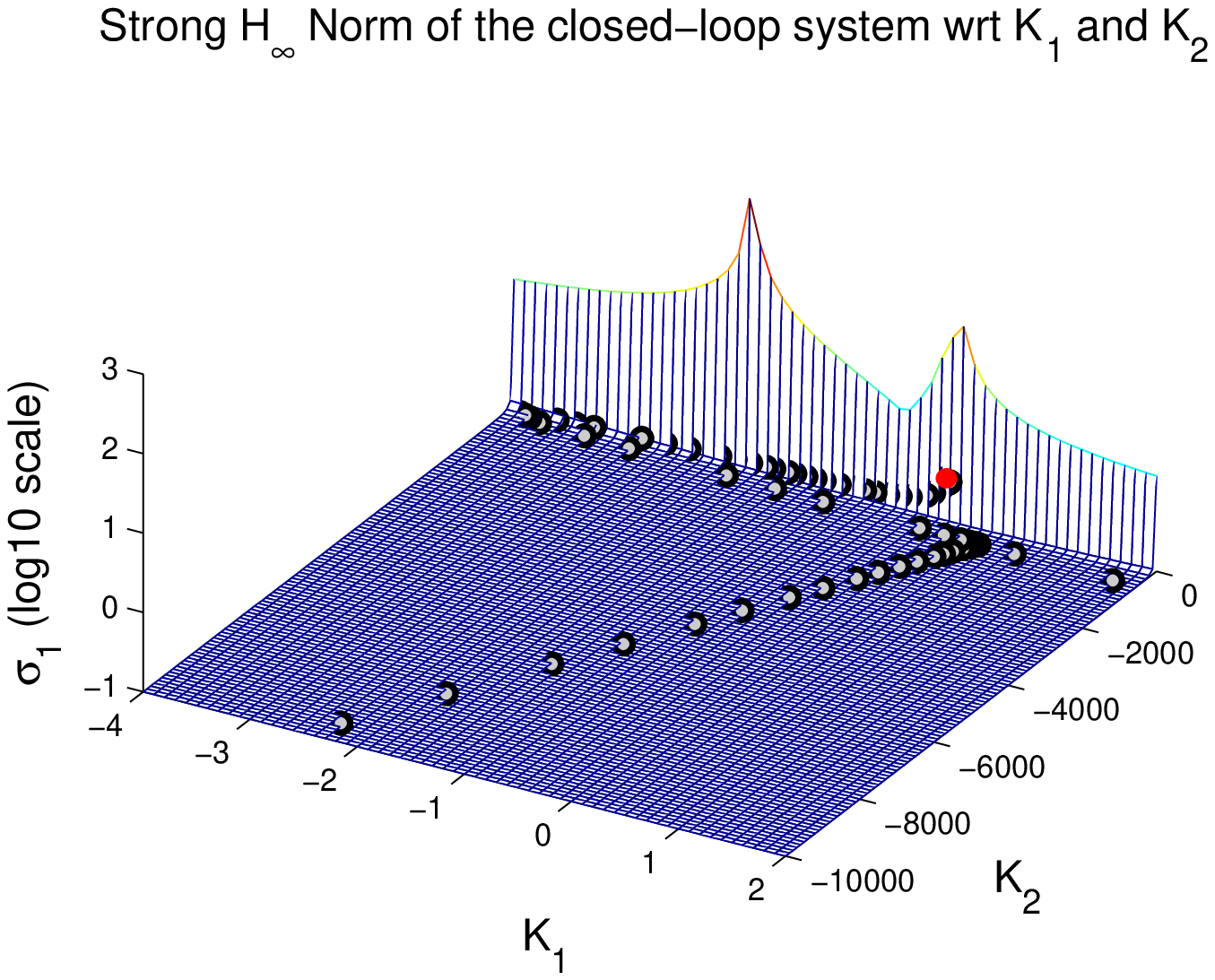}
\caption{\label{fig:tds2} The strong $\Hi$ norm of the closed-loop system in Example $4$ \cite{Fridman:2002:DESCRIPTOR} with respect to controller parameters, $K=\left[K_1\ K_2\right]$.}
\end{figure}

\begin{table}[h]
\begin{center}
\begin{tabular}{lccl}
  \hline
  \hline
  Problem & Other Methods & Results & Computed Controller\\
  \hline
  Ex. $3$, \cite{fridman} & $2.4$, \cite{fridman} & $3.7654$ & {\small $\left[-8.6961\right]$}\vspace{.5mm}\\
     &  & $1.2618$ & {\small $\left[\begin{array}{c|c}
                                       -7.1827 & -37.3389 \\
                                       \hline
                                       18.6767 & 90.4893
                                     \end{array}
   \right]$}\vspace{.5mm}\\
     &  & $1.2428$ & {\small $\left[\begin{array}{cc|c}
                                       -2.6837 & -15.1028 & -6.2101 \\
                                        0.3607 &   1.2086 &  3.6959\\
                                       \hline
                                        0.1379 &  -3.9720 & 10.4548
                                     \end{array}
   \right]$}\vspace{.5mm}\\
  \hline \\[-6pt]
  Ex. $4$, \cite{fridman} & $0.8600$, \cite{Fridman:2002:DESCRIPTOR} & $0.1617$ & {\small $\left[-16.1692\right]$}\vspace{.5mm}\\
  ($h=0.999$) & $11$, \cite{fridman}  & & \vspace{.5mm}\\
  \hline \\[-6pt]
  Ex. $4$, \cite{fridman} & $20$, \cite{fridman} & $0.1617$ & {\small $\left[-16.1692\right]$}\vspace{.5mm}\\
  ($h=1.28$) &   & & \vspace{.5mm}\\
  \hline
  \hline
\end{tabular}
\end{center}
\caption{The achieved $\Hi$ performances by output-feedback controllers} \label{table:output}
\end{table}

The designed controllers for the examples in Table~\ref{table:output} have a state-feedback-observer structure. The observer is a time-delay system and estimates the states of the original plant. Example $3$ is described by a DDAE and Example $4$ is a retarded time-delay system. The strong $\Hi$ norms of the closed-loop system in Example $4$ of \cite{fridman} for $h=0.999$ and $h=1.28$ with respect to controller parameters are given in Figure~\ref{fig:tds3} and \ref{fig:tds4}. The optimal controller gain does not change for two different delays. In both cases, the optimization method reaches the optimal value.

\begin{figure}[!h]
    \begin{minipage}[t]{0.45\textwidth}
        \vspace{0pt}\raggedright
        \includegraphics[width=\linewidth]{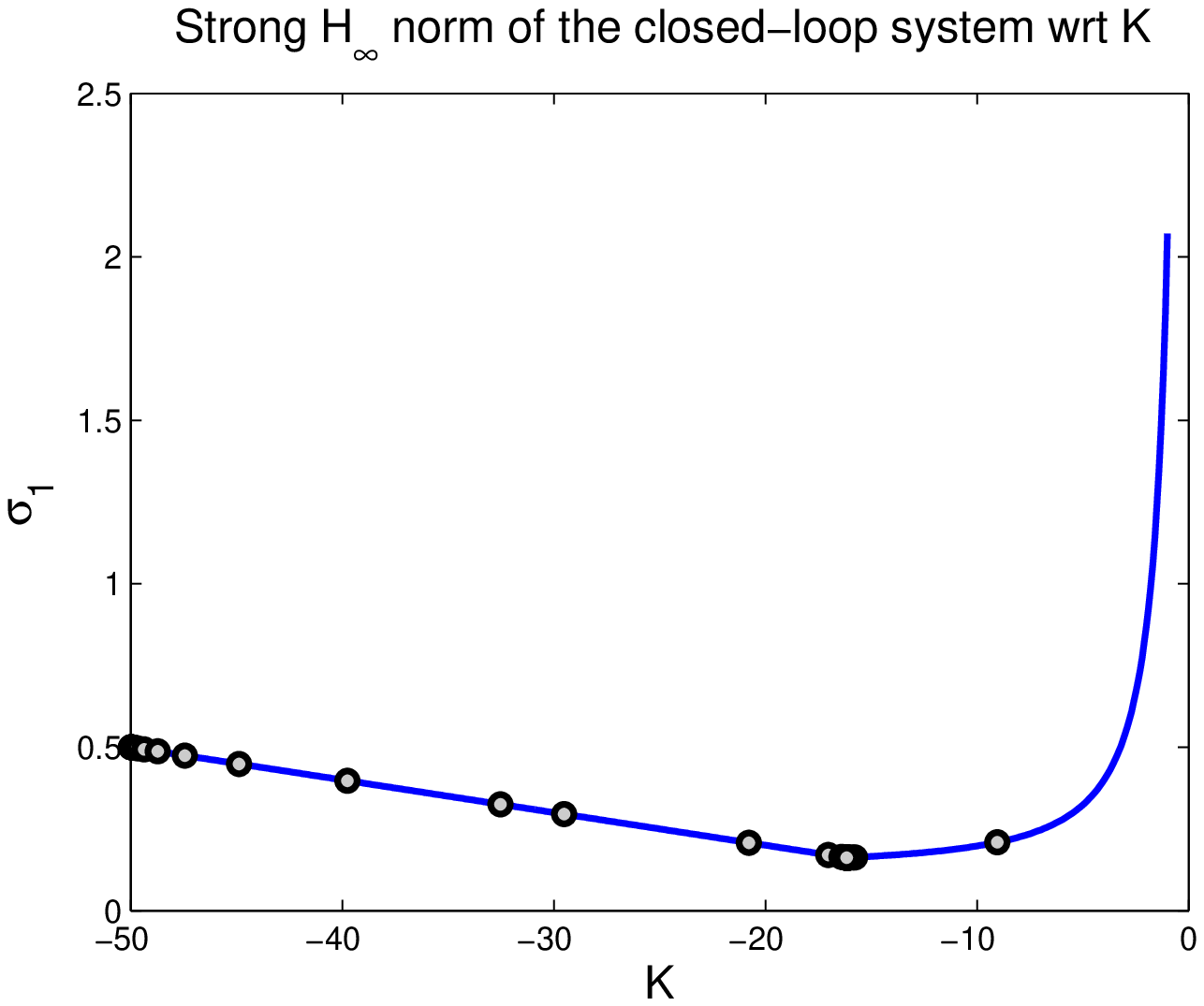}
        \caption{\label{fig:tds3} The strong $\Hi$ norm of the closed-loop system in Example $4$ ($h=0.999$) \cite{fridman} with respect to the controller parameter.}
   \end{minipage}
\hfill
    \begin{minipage}[t]{0.45\textwidth}
        \vspace{0pt}
        \includegraphics[width=\linewidth]{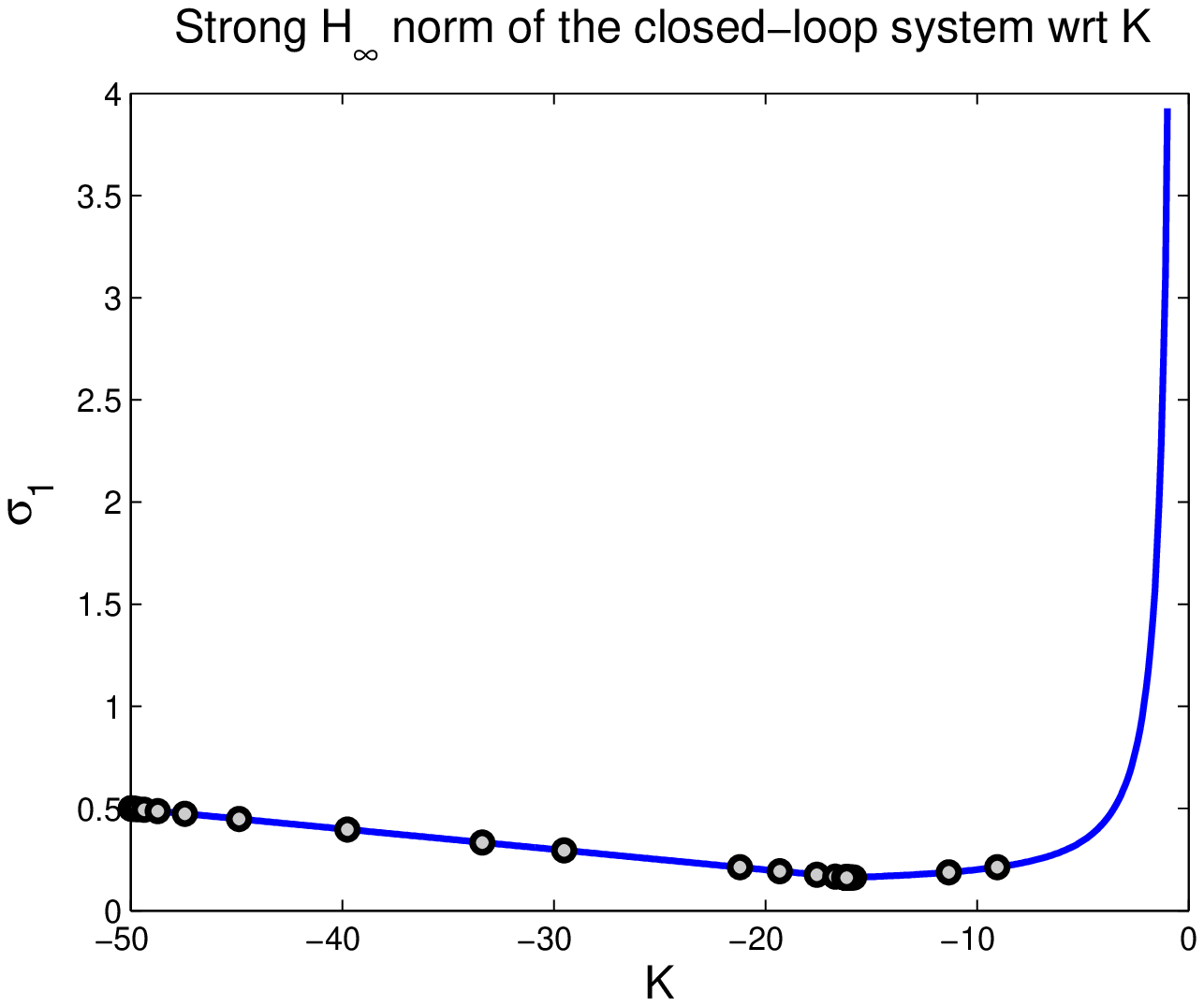}
        \caption{\label{fig:tds4} The strong $\Hi$ norm of the closed-loop system in Example $4$ ($h=1.28$) \cite{fridman} with respect to the controller parameter.}
        \end{minipage}
\end{figure}

In \cite{robustpaper} the robust stabilization of a time-delay system of retarded type by a static state-feedback controller is addressed. This problem is formulated as a $\Hi$ synthesis problem. The approximate $\Hi$ norm is computed using a frequency grid. The minimization is performed by a continuation approach where the highest peak values in the singular value plot are monitored. We applied our method on the numerical problem in Section $5$ of this reference and obtained a similar result. The optimized $\Hi$ norm is $3.3145$ and the corresponding state-feedback controller is given by
\[
K=\left[ \begin{array}{ccc}0.7763 & 1.1119 & 0.5433 \end{array} \right].
\]

In \cite{quasidirect} a state-feedback controller is designed for time-delay systems based on quasi-direct pole placement. This approach allows to assign a number of fixed right-most poles while shifting the remaining part of the spectrum as far to the left as possible. We designed a static $\Hi$ controller for the experimental heat transfer set-up described in Section~3 of \cite{quasidirect}. This is an $11^{\textrm{th}}$-order retarded time-delay system with $5$ state  delays and $1$ input delay,
%\[
%\left\{\begin{eqnarray*}
%\dot{x}(t)&=&A_0x(t)+\sum_{i=1}^5 A_i x(t-\tau_i)+w(t)+Bu(t-\tau_6), \\
%z(t)&=&x(t),\quad y(t)=x(t)
%\end{eqnarray*}\right\}
%\]

\[
\left\{\begin{array}{l}
\dot{x}(t)=A_0x(t)+\sum_{i=1}^5 A_i x(t-\tau_i)+w(t)+Bu(t-\tau_6), \\
z(t)=x(t),\quad y(t)=x(t)
\end{array}\right.
\] where system matrices and delays, $A_0$, $A_i$ and $\tau_i$ for $i=0,\ldots,5$ are given in \cite{quasidirect}. The performance channels are set to identity matrices. The controller parameters of four static controllers designed using quasi-direct pole-placement are given in Table $1$ of \cite{quasidirect}. The closed-loop strong $\Hi$ norms with these controllers are $779.1600$, $1881.3944$, $1155.7140$, $2113.8085$. We achieved a minimal closed-loop strong $\Hi$ norm $386.3491$ by a static $\Hi$ controller, $u(t)=K x(t)$, where
\begin{multline*}
K=[-1.3414\ \ -5.7544\ \ 1.0440\ \ 0.5181\ \ -29.9649\ \ -5.0182\ \    \\
-12.4284\ \ 0.6694\ \ 4.7125\ \ -23.6380\ \ 2.3902].
\end{multline*}

Finally, in~\cite{bfgbookchapter} a direct optimization approach is applied to the design of fixed-order $\Hi$ controllers for a class of retarded time-delay systems where the controller has no feedthrough term. The second example in \cite{bfgbookchapter} is a $4^\textrm{th}$-order time-delay system with $4$ delays. The system is stable and its $\Hi$ norm is $1.3907$. In Table~\ref{table:bfgbook}, we present our results for different controller orders $n_K$ for this example, without any additional restriction on the controller.

\begin{table}[h]
\begin{center}
\begin{tabular}{ccl}
  \hline
  \hline
  $n_K$  & Results & Computed Controller\\
  \hline
  1 & $1.2513$ & {\small $\left[\begin{array}{c|c}
                                       -0.3068 & 0.9590 \\
                                       \hline
                                       0.0166 & 0.0186
                                     \end{array}
   \right]$}\vspace{.5mm}\\

  2 & $1.2508$ & {\small $\left[\begin{array}{cc|c}
                                       -0.0959 & -0.0624 & -0.0982\\
                                       -0.0024 & -0.1984 &  0.0883 \\
                                       \hline
                                       -0.0756 & 0.0347 & 0.0234
                                     \end{array}
   \right]$}\vspace{.5mm}\\

  3 & $1.2493$ & {\small $\left[\begin{array}{ccc|c}
   -0.0861 &  -0.0673 &  -0.0953 & -0.0519 \\
    0.0046 &  -0.2170 &  -0.0233 &  0.1083 \\
   -0.0016 &   0.0010 &  -0.2973 &  0.1995 \\
   \hline
   -0.1734 &  -0.1040 &  -0.0475 &  0.0362
                                     \end{array}
   \right]$}\vspace{.5mm}\\
  \hline
  \hline
\end{tabular}
\end{center}
\caption{The achieved $\Hi$ performances for the time-delay system in \cite{bfgbookchapter} by dynamic controllers} \label{table:bfgbook}
\end{table}

\subsection{Benchmark results} \label{sec:bencmarks}

In Table~\ref{table:bench}, we present the results of benchmarking of our code with $5$  additional problems. The plants are retarded time-delay systems of the form (\ref{plant}). The second column shows the size of matrices $A_i$, $n$, and the total number of time-delays in the plant, $m$. The third column gives order $n_K$ of the (finite-dimensional) controller. The first line for each plant displays the closed-loop strong $\Hi$ norm when there is no controller (excepting the fourth plant which is unstable). The fourth and fifth columns contain the optimized strong $\Hi$ norm of the closed-loop system and the corresponding controller.

\begin{table}[!h]
\begin{center}
\begin{tabular}{ccccl}
  \hline
  \hline
  Plants & (n,m) & $n_K$ & Results & Computed Controller\\
  \hline
  \#1 & (2,2) & $-$ & $44.7086$  & $-$\vspace{.5mm}\\
      &       & $0$ & $5.1499$  &{\small $\left[-2.8858\right]$}\vspace{.5mm}\\
      &       & $1$ & $5.1037$  &{\small $\left[
      \begin{array}{c|c}
        -3.8529 & 1.4636 \\
        \hline
        -1.8703 & -2.0389
      \end{array}
      \right]$}
      \vspace{.5mm}\\
      &       & $2$ & $5.1029$  &{\small $\left[
      \begin{array}{cc|c}
      -4.6132  &  6.5495 & -2.2860 \\
      -3.1847  & -4.7981 &  2.4295 \\
      \hline
      -3.9177  &  1.0307 &  -2.7309 \\
      \end{array}
      \right]$}
      \vspace{.5mm}\\
  \hline
  \#2 & (3,5) & $-$ & $17.3595$  & $-$\vspace{.5mm}\\
      &       & $0$ & $1.9604$  &{\small $\left[0.0641\right]$}\vspace{.5mm}\\
      &       & $1$ & $1.9004$  &{\small $\left[
      \begin{array}{c|c}
        -0.1140 & 1.9984 \\
        \hline
        -0.0083 & 0.0626
      \end{array}
      \right]$}
      \vspace{.5mm}\\
      &       & $2$ & $1.8695$  &{\small $\left[
      \begin{array}{cc|c}
      -0.4458  & -0.0416 & 1.9645 \\
      -0.4998  & -0.4274 & 0.9799 \\
      \hline
      -0.0223  &  0.0125 & 0.0653 \\
      \end{array}
      \right]$}
      \vspace{.5mm}\\
  \hline
  \#3 & (4,7) & $-$ & $22.5669$  & $-$\vspace{.5mm}\\
      &       & $0$ & $7.9522$  &{\small $\left[
      \begin{array}{c}
      -0.4435 \\
      -2.5784
      \end{array}
      \right]$}\vspace{.5mm}\\
      &       & $1$ & $7.9220$  &{\small $\left[
      \begin{array}{c|c}
        -2.3222 & 0.3329 \\
        \hline
        0.5662  & -0.4290 \\
        -0.9636 & -2.5278
      \end{array}
      \right]$}
      \vspace{.5mm}\\
      &       & $2$ & $7.9071$  &{\small $\left[
      \begin{array}{cc|c}
        -2.3261 &  0.3181 &  0.2772 \\
        -0.7214 & -1.4068 &  0.1320 \\
        \hline
         0.5882 & -0.2219 & -0.4312 \\
        -0.9382 &  0.3012 & -2.5353
      \end{array}
      \right]$}
      \vspace{.5mm}\\
  \hline
  \#4 & (6,2) & $0$ & $165.0013$  &{\small $\left[-2\right]$}\vspace{.5mm}\\
      &       & $0$ & $0.4199$  &{\small $\left[-3778.5600\right]$}\vspace{.5mm}\\
      &       & $1$ & $0.2275$  &{\small $\left[
      \begin{array}{c|c}
        -121.3792 & -1.4191 \\
        \hline
         -27.6931 & -13.8822
      \end{array}
      \right]$}
      \vspace{.5mm}\\
      &       & $2$ & $0.1081$  &{\small $\left[
      \begin{array}{cc|c}
        -2.3590 &  9.2505 & -11.2683 \\
        -1.5130 & -3.6363 &   5.8830 \\
        \hline
        -5.9382 &  7.9907 & -23.7252
      \end{array}
      \right]$}
      \vspace{.5mm}\\
  \hline
  \#5 & (8,6) & $-$ & $1.4291$  &$-$ \vspace{.5mm}\\
      &       & $0$ & $0.4751$  &{\small $\left[-1.9792\right]$}\vspace{.5mm}\\
      &       & $1$ & $0.2818$  &{\small $\left[
      \begin{array}{c|c}
        -1.3299 & 1.2660 \\
        \hline
        -6.1379 & -1.0250
      \end{array}
      \right]$}
      \vspace{.5mm}\\
      &       & $2$ & $0.2809$  &{\small $\left[
      \begin{array}{cc|c}
        -1.7327 &   1.9221 &  -4.0733 \\
        -2.6273 &  -2.0103 &  -2.7046 \\
        \hline
         2.6669 &  -3.2341 &  -1.4181
      \end{array}
      \right]$}
      \vspace{.5mm}\\
  \hline
  \hline
\end{tabular}
\end{center}
\caption{The achieved $\Hi$ performances for benchmark problems} \label{table:bench}
\end{table}

The problem data for the above benchmark examples and a MATLAB implementation of our code for the strong $\Hi$ controller design are available at the website
\begin{verbatim}
http://twr.cs.kuleuven.be/research/software/delay-control/hinfopt/.
\end{verbatim}

\section{Conclusions} \label{sec:conc}

%We considered the fixed-order $\Hi$ controller design for delay differential %algebraic systems. We illustrated how to obtain these systems from different plant and controller configurations without any elimination techniques. We showed the sensitivity of the $\Hi$ norm to small delay changes, defined the strong $\Hi$ norm for these systems robust to delays, analyzed its properties and gave a numerical method to compute this norm. Using the given numerical method for the strong $\Hi$ norm computation and the computed derivatives with respect to controller parameters, we designed robust controllers based on non-smooth, non-convex optimization techniques allowing the user to choose the controller order as desired.

We considered the fixed-order/fixed-structure $\Hi$ controller design problem for delay differential algebraic systems. The main contributions are as follows.
\begin{enumerate}
\item We show that a very broad class of interconnected systems can be brought in the standard form (\ref{system}) in a systematic way. Input/output delays and direct feedthrough terms can be dealt with by introducing slack variables. The dependence of the closed-loop matrices on the controller parameters always remains linear.
\item We demonstrated the sensitivity of the $\Hi$ norm w.r.t.~small delay perturbations and  introduced the \emph{strong $\Hi$ norm} for DDAEs, inline with the notion of strong stability, and we analyzed its properties.
\item We presented a predictor-corrector algorithm for the (strong) $\Hi$ norm computation of DDAEs.
\item Based on the numerical algorithm for the strong $\Hi$ norm and its gradient computation with respect to controller parameters, we applied non-smooth, non-convex optimization methods for designing controllers with a fixed order or structure.
\end{enumerate}
The presented approach has been validated by numerical examples. An implementation of the algorithms is available from
\begin{verbatim}
http://twr.cs.kuleuven.be/research/software/delay-control/hinfopt/.
\end{verbatim}

\section*{Acknowledgements}
This article present results of the Belgian Programme on Interuniversity Poles of Attraction, initiated by the
Belgian State, Prime Minister's Office for Science, Technology and Culture, of the Optimization in Engineering Centre OPTEC, and of the project STRT1-09/33 of the K.U.Leuven Research Council.

\bibliographystyle{plain}
\bibliography{referentielijst,otherref}

\appendix

\section{Some technical lemmas}\label{sec:Appendix2}

\begin{lemma}\label{propconverge2}
For all $\gamma>0$, there exist numbers $\epsilon>0$ and  $\Omega>0$ such that
\[
\sigma_{1}\left(T(j\w,\vec r)-T_a(j\w,\vec r)\right)<\gamma
\]
 fir all $\w>\Omega$ and $\vec r\in\mathcal{B}(\vec\tau,\epsilon)\cap(\mathbb{R}^+)^m$.
\end{lemma}
\begin{proof}
   The uniformity of the bound $\gamma$ w.r.t.~small delay perturbations  stems from the fact that the bound (\ref{normA22}) is a continuous function of the delays $\vec \tau$ at their nominal values. The latter is implied  by  the \emph{strong} stability assumption (Assumption~\ref{assumption_sstab}).
\end{proof}

\begin{lemma} \label{lem:fininter}
Let $\xi>\interleave T_a(j\w,\vec \tau)\interleave_\infty$ hold.
Then there exist real numbers $\epsilon>0,\ \Omega>0$ and an integer $N$ such that for any $\vec r\in\mathcal{B}(\vec\tau,\epsilon)\cap(\mathbb{R}^+)^m$, the number of frequencies $\w^{(i)}$ such that
\begin{equation}
\sigma_k\left(T(j\w^{(i)},\vec r)\right)=\xi,
\end{equation}
for some $k\in\{1,\ldots,n\}$, is smaller then $N$, and, moreover, $|\w^{(i)}|<\Omega$.
\end{lemma}

\begin{proof}
 For any (fixed) value of $\xi>0$ and delays $\vec r$,  the relation
 \begin{equation}\label{omhold}
 \sigma_k\left(T(j\w),\vec r\right)=\xi
\end{equation}
 holds for some $\omega\in\mathbb{R}$ and $k\in\{1,\ldots,n\}$ if and only if $\lambda=j\omega$ is a zero of the function
\begin{equation}\label{ham}
\det\left(\left[\begin{array}{cc}
\lambda E-A_0-\sum_{i=1}^m A_i e^{-\lambda r_i} & -\frac{1}{\xi}B B^T \\
\frac{1}{\xi}CC^T & \lambda E^T+A_0^T +\sum_{i=1}^m A_i^T e^{\lambda r_i}
\end{array}\right]\right).
\end{equation}
This result is a variant of Lemma~2.1 of \cite{wimsimax} to which we refer for the proof.

Now take $\xi>\interleave T_a(j\w,\vec \tau)\interleave_\infty$. From Lemma~\ref{propconverge2}, and taking into account that $\interleave T_a(j\w,\vec \tau)\interleave_\infty$ does not depend on $\vec\tau$ (see Proposition~\ref{prop:Tasinfprop}) it follows that there exists numbers  $\epsilon>0$ and $\Omega>0$ such that all $\omega$ satisfying (\ref{omhold}) for some $k\in\{1,\ldots,n\}$ and $\vec r\in\mathcal{B}(\vec\tau,\epsilon)\cap(\mathbb{R}^+)^m$ also satisfy
$
|\omega|< \Omega.
$
This proves one statement. At the same time $\lambda=j\omega$ must be a zero of the analytic function (\ref{ham}). The other statement is due to the fact that an analytic function only has finitely many zeros in a compact set.
\end{proof}

\begin{lemma}\label{lem3ap}
The following implication holds
\[
\| T(j\w,\vec \tau)\|_{\infty} \leq  \interleave T_a(j\w,\vec \tau)\interleave_\infty \ \Rightarrow
\interleave T(j\w,\vec \tau)\interleave_{\infty} = \interleave T_a(j\w,\vec \tau)\interleave_\infty.
\]
\end{lemma}

\begin{proof}
For every $\epsilon>0$ there exist delays $\vec\tau_0$ and a frequency $\omega_0$ such that
\[
\|\vec\tau_0-\vec\tau\|<\epsilon/2,\ \ \ \sigma_{1}\left(T_a(j\omega_0,\vec\tau_0)\right)\geq \interleave T_a(j\w,\vec \tau)\interleave_\infty-\epsilon/2.
\]
In addition, there exist commensurate delays
\begin{equation}\label{comdelays}
\vec\tau_r=\left(\frac{n_1}{s},\ldots,\frac{n_m}{s}\right),
\end{equation}
with $(n_1,\ldots,n_m,s)\in\mathbb{N}^{m+1}$ such that
\[
\|\vec\tau_r-\vec\tau_0\|<\epsilon/2,\ \ \ \left|\sigma_{1}\left(T_a(j\omega_0,\vec\tau_r)\right)-
\sigma_{1}\left(T_a(j\omega_0,\vec\tau_0)\right)\right|\leq\epsilon/2.
\]
Thus, for all $\epsilon>0$ there exist commensurate delays (\ref{comdelays}) and a frequency $\omega_0$ satisfying
\[
 \|\vec\tau_r-\vec\tau\|<\epsilon,\ \
\sigma_{1}\left(T_a(j\omega_0,\vec\tau_r)\right)\geq \interleave T_a(j\w,\vec \tau)\interleave_\infty-\epsilon.
\]
From the fact that
\[
T_a(j\omega_0,\vec\tau_r)=T_a\left(j(\omega_0+2\pi s k) ,\vec\tau_r\right)
\]
for all $k\geq 1$ and Lemma~\ref{propconverge2}, we conclude that
\begin{equation}\label{leftright}
\interleave T(j\w,\vec \tau)\interleave_{\infty} \geq  \interleave T_a(j\w,\vec \tau)\interleave_\infty.
\end{equation}

Now take a level $\xi> \interleave T_a(j\w,\vec \tau)\interleave_\infty$, and let $\epsilon$ and $\Omega$ be determined by the assertion of Lemma~\ref{lem:fininter}. From the assumption $\| T(j\w,\vec \tau)\|_{\infty} \leq  \interleave T_a(j\w,\vec \tau)\interleave_\infty$  and the relation between (\ref{omhold}) and (\ref{ham}) it follows that the function (\ref{ham}) has no zeros on the imaginary axis for $\vec r=\vec\tau$. Because the function  (\ref{ham}) is analytic and all potential imaginary axis zeros have modulus smaller than $\Omega$ whenever $\vec r\in\mathcal{B}(\vec\tau,\epsilon)\cap(\mathbb{R}^+)^m$, we conclude that there exists a number $\epsilon_2>0$ such that the function (\ref{ham}) has no imaginary axis eigenvalues whenever $\vec r\in\mathcal{B}(\vec\tau,\epsilon_2)\cap(\mathbb{R}^+)^m$. Equivalently, $T(j\omega,\vec r)$ has no singular values equal to $\xi$ whenever $\vec r\in\mathcal{B}(\vec\tau,\epsilon_2)\cap(\mathbb{R}^+)^m$.  This proves that the  left and the right hand side of (\ref{leftright}) are equal.
\end{proof}

\section{Finite dimensional approximation} \label{sec:findimapp}

We start by reformulating the system (\ref{system}) as an infinite-dimensional linear system, inspired by \cite{curtainzwart}. When defining the Hilbert space $X:=\C^n\times \mathcal{L}_2([-\tau_{\max},0],\C^n)$ equipped with the
inner product
\[
<(y_0,y_1),(z_0,z_1)>_X=<y_0,z_0>_{\C^n}+<y_1,z_1>_{\mathcal{L}_2},
\] where $\tau_{\max}:=\max(\tau_1,\ldots,\tau_m)$, we can rewrite (\ref{system}) as
 \begin{eqnarray} \label{eq:ODEA}
\mathcal{E}\dot{z}(t)&=&\mathcal{A}z(t)+\mathcal{B}u(t),\\
\nonumber y(t)&=&\mathcal{C}z(t),
\end{eqnarray} where
\begin{multline} %\label{def:infgenA}
 \mathcal{D}(\mathcal{A}) =\{z=(z_1,z_0)\in X: z_1
\mathrm{\ is\ absolutely\  continuous} \\
 \mathrm{on\ } [-\tau_{\max},0], \frac{dz_1}{d\theta}\in\mathcal{C}([-\tau_{\max},0],\C^n), z_0=z_1(0)\}, \\
\end{multline}
\begin{eqnarray}
\nonumber \mathcal{E}z&=&\left(\begin{array}{c}
         z_1\\
         Ez_0
         \end{array}\right), \hspace{1.4cm}
          \mathcal{A}z=\left(\begin{array}{c}
         \frac{dz_1}{d\theta}(.)\\
         A_{0}z_0+ \sum_{i=1}^m A_i z_1(-\tau_i)
         \end{array}\right), z\in\mathcal{D}(\mathcal{A}), \\
\nonumber \mathcal{B}u&=&\left(\begin{array}{c}
        0\\
        Bu
        \end{array}
    \right), u\in\C^{n\times n_u},\  \ \mathcal{C}z=Cz_0,\  z\in X.
\end{eqnarray}
The connection between (\ref{system})
and (\ref{eq:ODEA}) is that $z_0(t)\equiv x(t)$,
$z_1(t)\equiv x(t+\theta), \theta\in[-\tau_{\max},0]$.

Next, we discretize the infinite-dimensional system
(\ref{eq:ODEA}). We use a spectral method, as in
\cite{breda,breda:nonlocal}. Given a positive integer $N$,
we consider a mesh $\Omega_N$ of $N+1$ distinct points in
the interval $[-\tau_{\max},\ 0]$,
\begin{equation}\label{defmesh}
\Omega_N=\left\{\theta_{N,i},\ i=-N,\ldots,0\right\},
\end{equation}
where we assume that $\theta_{N,0}=0$. With the Lagrange
polynomials $l_{N,k}$ defined as real valued polynomials
of degree $N$ satisfying
\[
l_{N,k}(\theta_{N,i})=\left\{\begin{array}{ll}1 & i=k\\
0 & i\neq k
\end{array}\right.
\]
for $i,k\in\{-N,\ldots,0\}$, we can, similarly as in \cite{breda}, approximate (\ref{eq:ODEA}) and, hence, (\ref{system}), by the
finite-dimensional system:
\begin{eqnarray}\label{finsystem}
{\bf E_N}\dot z(t)&=&{\bf A_N} z(t)+{\bf B_N} u(t),\ z(t)\in\R^{(N+1)n\times 1} \\
y(t)&=&{\bf C_N}z(t)
\end{eqnarray}
where
\begin{eqnarray}
% A_N
\label{ANBN} {\bf A_N}&=&\left[\begin{array}{llll}
d_{-N,-N} I_n &\hdots & d_{-N,-1} I_n & d_{-N,0} I_n \\
\vdots & &  \vdots & \vdots \\
d_{-1,-N} I_n &\hdots & d_{-1,-1} I_n  & d_{-1,0} I_n \\
\Gamma_{-N} & \hdots & \Gamma_{-1} & \Gamma_0
\end{array}\right],\
% B_N
{\bf B_N}=\left[\begin{array}{l}
0 \\
\vdots \\
0 \\
B \\
\end{array}\right] \\
{\bf E_N}&=&
% E_N
\left[\begin{array}{ll}
I_{nN} & 0_{nN\times n} \\
0_{n\times nN}  & E \\
\end{array}\right],\
% C_N
{\bf C_N}=\left[\begin{array}{llll}
0 & \ldots & 0 & C
\end{array}\right]
\end{eqnarray} and
\[
\begin{array}{lll}
\Gamma_0&=& A_0+\sum_{l=1}^m A_l l_{N,0}(-\tau_l),\\
\Gamma_{k}&=&\sum_{l=1}^m A_l l_{N,k}(-\tau_l),\ \ \
k\in\{-N,\ldots,-1\}, \\
d_{i,k}&=&l^{\prime}_{N,k}(\theta_{N,i}) ,\ \ \ \
i,k\in\{-N,\ldots,0\}.
\end{array}
\]

The transfer function of (\ref{finsystem}) is given by (\ref{finite2}).
Using the arguments as spelled out in \cite{suataut,wimsimax} it can be shown that the effect of the approximation of
(\ref{T}) by (\ref{finite2}) can be interpreted as the effect of approximating  the exponential functions in (\ref{T}) by rational functions. This interpretation can be used to estimate the frequency interval where the main peaks in the singular value plot occur (see \cite[\S4.3]{wimsimax}).

\section{Numerical data in Section \ref{sec:collection}} \label{sec:numdata72}
The state-space equations of the numerical examples in Section \ref{sec:collection} are as follows.

\subsection{Example $4$ in \cite{Fridman:2002:DESCRIPTOR}}
\begin{align*}
\dot{x}_1(t)&=-x_1(t-0.999)-x_2(t-0.999)+w(t), \\
\dot{x}_2(t)&=x_2(t)-0.9x_2(t-0.999)+w(t)+u(t), \\
z_1(t)&=x_2(t),\quad z_2(t)=0.1u(t),\quad  y_1(t)=x_1(t),\quad y_2(t)=x_2(t).
\end{align*}
\subsection{Example $1$ in \cite{Fridman1998SCL}}
\begin{align*}
\dot{x}_1(t)&=2x_1(t)+x_2(t)-x_1(t-0.1)-0.5w(t)+3u(t), \\
\dot{x}_2(t)&=-x_2(t)-x_1(t-0.1)+x_2(t-0.1)+w(t)+u(t), \\
z_1(t)&=x_1(t)-0.5x_2(t),\quad z_2(t)=u(t),\quad  y_1(t)=x_1(t),\quad y_2(t)=x_2(t).
\end{align*}
\subsection{Example $2$ in \cite{fridman}}
\begin{align*}
\dot{x}_1(t)&=-x_1(t-1.2)+w(t)-0.5u(t),\quad 0=x_1(t-1.2)-x_2(t-1.2)+w(t)+u(t), \\
z_1(t)&=x_1(t)+0.2x_2(t)+0.1u(t),\quad  y_1(t)=x_1(t),\quad y_2(t)=x_2(t).
\end{align*}
\subsection{Example $3$ in \cite{fridman}}
\begin{align*}
\dot{x}_1(t)&=-x_1(t-1.2)+w_1(t),\quad 0=x_2(t)+x_1(t-1.2)-x_2(t-1.2)+w_1(t)+u(t), \\
z_1(t)&=x_1(t)+0.2x_2(t)+0.1u(t),\quad  y_1(t)=x_1(t)+0.1w_2(t).
\end{align*}
\subsection{Example $4$ in \cite{fridman}}
\begin{align*}
\dot{x}_1(t)&=-x_1(t-h)-x_2(t-h)+w_1(t),\quad \dot{x}_2(t)=x_2(t)-0.9x_2(t-h)+w_1(t)+u(t), \\
z_1(t)&=x_2(t),\quad z_2(t)=0.1u(t),\quad y_1(t)=x_2(t)+0.1w_2(t).
\end{align*}
\subsection{Example in \cite{robustpaper}}
\[
\dot{x}(t)=Ax(t)+Bw(t)+Bu(t-5),\quad z(t)=I_3x(t),\quad y(t)=I_3x(t)
\] where
{\small
\[
\begin{array}{cc}
A=\left(\begin{array}{ccc}
  -0.08 & -0.03 & 0.2 \\
    0.2 & -0.04 & -0.005 \\
  -0.06 & 0.2   & -0.07
\end{array}\right), &
B=\left(\begin{array}{c}
  -0.1 \\
  -0.2 \\
  0.1
\end{array}\right).
\end{array}
\]}
\subsection{Example in \cite{quasidirect}}
\[
\dot{x}(t)=A_0+\sum_{i=1}^5 A_ix(t-h_i)+I_{11}w(t)+Bu(t-7),\quad z(t)=I_{11}x(t),\quad y(t)=I_{11}x(t)
\] where the time-delays are $h_1=3$, $h_2=5$, $h_3=15$, $h_4=23$, $h_5=29$. The system matrices $A_i$ for $i=0,\ldots,6$ are real-valued $11\times11$ matrices. The $(k,l)$ non-zero element of $A_i$th matrix is denoted by $A_i^{(k,l)}$ and the numerical values are $A_0^{(1,1),(3,3),(5,5),(9,9)}=-0.2$, $A_0^{(4,7),(4,8),(8,3),(8,4)}=0.1417$, $A_0^{(2,2)}=-0.04$, $A_0^{(6,6)}=-0.0588$, $A_0^{(6,6)}=-1$, $A_0^{(10,10)}=-0.0667$, $A_0^{(4,3)}=A_0^{8,7}=0.1917$,
$A_0^{(4,4)}=A_0^{(8,8)}=-0.04$, $A_1^{(5,4)}=0.195$, $A_2^{(3,2)}=0.1966$, $A_2^{(6,5)}=0.0529$, $A_2^{(9,8)}=0.194$, $A_2^{(10,9)}=0.0613$, $A_3^{1,6}=0.1946$, $A_4^{2,1}=0.0384$, $A_5^{7,7}=-0.0159$.

\subsection{Example $2$ in \cite{bfgbookchapter}}

{\scriptsize
\begin{align*}
\dot{x}(t)&=\left(
\begin{array}{cccc}
   -4.4656  &  -0.4271 &   0.4427 &  -0.1854 \\
   -0.8601  & -5.6257  &  0.8577  & -0.5210 \\
    0.9001  & -0.7177  & -6.5358  &  0.0417 \\
   -0.6836  &  0.0242  &  0.4997  & -3.5618
\end{array}
\right)x(t)+\left(
\begin{array}{cccc}
    0.6848 &   -0.0618 &   0.5399 &   0.5057 \\
    0.3259 & -0.3810  &  0.6592  & -0.0066 \\
    0.6325 &   0.3752  &  0.4122  &  0.7303 \\
    0.5878  &  0.9737  &  0.1907 &  -0.8639
\end{array}
\right) \\
&x(t-3.2)+
\left(
\begin{array}{cccc}
    0.9371 &  -0.7859  &  0.1332 &   0.7429 \\
   -0.8025 &   0.4483  &  0.6226  &  0.0152 \\
    0.0940  &  0.2274  &  0.1536  &  0.5776 \\
   -0.1941  &  0.5659  &  0.8881 &  -0.0539
\end{array}\right)x(t-3.4)+
\left(
\begin{array}{cc}
    1    &     0 \\
   -1.6  & 1 \\
         0   &      0 \\
         0    &     0
\end{array}
\right)w(t) \\
&+
\left(
\begin{array}{cccc}
    0.6576 &  -0.8543 &  -0.3460 &   0.6415 \\
   -0.3550 &   0.5024 &   0.6081 &   0.9038 \\
    0.9523  &  0.6624 &   0.0765 &  -0.8475 \\
   -0.4436  &  0.8447 &  -0.0734 &   0.4173 \\
\end{array}
\right)x(t-3.9)+
\left(
\begin{array}{cc}
    0.2 \\
   -1 \\
    0.1 \\
   -0.4 \\
\end{array}
\right)u(t-0.2) \\
z(t)&=
\left(
\begin{array}{cccc}
     1 &    0 &    0 &   -1 \\
     0 &   -1 &    1 &    0
\end{array}
\right)x(t)
+
\left(
\begin{array}{cc}
    0.1 & 1 \\
   -1 &  0.2
\end{array}
\right) w(t)
+
\left(
\begin{array}{c}
    1 \\
   -1
\end{array}
\right) u(t) \\
y(t)&=
\left(
\begin{array}{cccc}
     1 &    0 &    -1 &   0
\end{array}
\right)x(t)
+
\left(
\begin{array}{cc}
    -2 & 0.1
\end{array}
\right) w(t)
+
0.4 u(t-0.2)
\end{align*}
}
\end{document}